%% file: proof_rhie_arxiv.tex
\documentclass[a4paper, 11pt]{article}

\usepackage[arxiv]{optional}
\include{texdef}

\title{Sharp parameter bounds for certain maximal point
lenses} \date{April 3, 2014}

\author{Robert Luce, Olivier S\`{e}te, J\"org Liesen\footnote{TU
Berlin, Institute of Mathematics, MA 4-5, Stra{\ss}e des 17. Juni 136,
10623 Berlin, Germany (\texttt{\{luce,sete,liesen\}@math.tu-berlin.de})}}

\begin{document}
\maketitle

\begin{abstract}
\input{abstract}

\end{abstract}

\input{proof_rhie_mainbody}

\paragraph*{Ackowledgements}
Robert Luce's work is supported by Deutsche Forschungsgemeinschaft,
Cluster of Excellence ``UniCat''.

\begin{appendix}
\input{proof_rhie_appendix}
\end{appendix}

\bibliographystyle{plain}
\bibliography{proof_rhie}

\end{document}

%% file: texdef.tex
\usepackage[english]{babel}
\usepackage[T1]{fontenc}
\usepackage{amsmath}
\opt{arxiv}{%
\usepackage{amsthm}
}%
\usepackage{amsfonts}
\usepackage{amssymb}
\usepackage[final]{graphicx}
\usepackage{paralist}
\opt{draft}{%
\usepackage{refcheck}
}

\opt{arxiv}{%
\theoremstyle{plain}
\newtheorem{thm}{Theorem}[section]
\newtheorem{prop}[thm]{Proposition}
\newtheorem*{prop*}{Proposition}
\newtheorem{lem}[thm]{Lemma}
\newtheorem*{lem*}{Lemma}
\newtheorem{cor}[thm]{Corollary}
\newtheorem*{cor*}{Corollary}

\newtheorem*{example*}{Example}
\newtheorem{conject}[thm]{Conjecture}
\newtheorem*{conject*}{Conjecture}

\theoremstyle{definition}
\newtheorem{defn}[thm]{Definition}
\newtheorem*{defn*}{Definition}
\newtheorem{rem}[thm]{Remark}
\newtheorem*{rem*}{Remark}
}%
\opt{springer}{%

}%

\newcommand{\eps}{\varepsilon}

\DeclareMathOperator{\crit}{cr}

\newcommand{\Msun}{\textrm{M}_{\odot}}

\opt{springer}{%
\newcommand{\eop}{\qed}
}
\opt{arxiv}{%
\newcommand{\eop}{}
}

\newcommand{\R}{\ensuremath{\mathbb{R}}}

\newcommand{\defby}{\mathrel{\mathop:}=}

\newcommand{\bydef}{=\mathrel{\mathop:}}

\newcommand{\coloneq}{\mathrel{\mathop:}=}

\newcommand{\conj}[1]{\overline{#1}}

\newcommand{\rpert}{\ensuremath{R_{\varepsilon}}}


%% file: abstract.tex
Starting from an $n$-point circular gravitational lens having $3n+1$
images, Rhie (2003) used a perturbation argument to construct an $(n+1)$-point
lens producing $5n$ images. In this work we give a concise proof of
Rhie's result, and we extend the range of parameters in Rhie's model
for which maximal lensing occurs.

We also study a slightly different construction given by Bayer and Dyer (2007)
arising from the $(3n+1)$-point lens. In particular, we extend their results and give
sharp parameter bounds for their lens model.  
By a substitution of variables and parameters we show
that both models are equivalent in a certain sense.

%% file: proof_rhie_mainbody.tex
\section{Introduction}

Gravitational lensing is the deflection of light due to masses located
between a light source and an observer, which can produce multiple
images of the light source.  Standard references describing
observations, theory and applications of gravitational lensing include
~\cite{PettersLevineWambsganss2001,SchneiderEhlersFalco1999,Wambsganss1998}.

In this paper we focus on an aspect in the theory of gravitational microlensing,
namely the question of \emph{maximal lensing}, from a mathematical point of
view.  We consider the case of a single lens plane containing $n$
point masses,
and no external shear.  Further we assume that the light source is located on the
optical axis from the observer through the origin of the lens plane (the center
of mass of the point masses).  As described, e.g.,
in~\cite{MaoPettersWitt1999,KhavinsonNeumann:2006},
the lens equation
can be written in this case (using complex numbers) as
\begin{equation*}
0 = z_s = z - \sum_{k=1}^n \tfrac{m_k}{\conj{z}-\conj{z}_k},
\end{equation*}
where $z_s$ is the position of the light source projected on the lens plane,
and $m_k$ is the mass of the deflector at the point $z_k$ in the lens plane.
We point out that several other lens models exist; see, e.g., the recent
surveys~\cite{Petters2010,PettersWerner2010} or~\cite{DalalRabin2001} 
(in particular Table~1 which gives an overview of six different models).

Mathematically,
the number of lensed images of the source in the above lens equation
is equal to the number of (complex) solutions of
\begin{equation*}
\sum_{k=1}^n \tfrac{m_k}{z-z_k}=\conj{z}.
\end{equation*}
An important special case of this equation, which we consider in this
paper, is the one of $n\geq 3$ equal masses that are located at the vertices of
a regular polygon in the lens plane with total mass normalized to unity, i.e.,
we set $m_k = \frac{1}{n}$ and $z_k = r e^{i \frac{2 k \pi}{n}}$ for some real
number $r>0$ and $k = 0, 1, \ldots, n-1$.  The corresponding lens equation,
\begin{equation}
    \label{eqn:unpert_rat_fun}
    R(z)=\conj{z},\quad\mbox{where}\quad R(z) \defby
    \tfrac{z^{n-1}}{z^n - r^n},
\end{equation}
has been studied in several publications;
see~\cite[section~2]{PettersWerner2010} or~\cite[sections 2,3]{FassnachtKeetonKhavinson2009} for
overviews with many pointers to the literature. In particular,
Mao, Petters, and Witt~\cite{MaoPettersWitt1999} showed that this
equation
has $3n+1$ solutions if
\begin{equation*}
r<r_{\crit} \defby \big( \big( \tfrac{n-2}{n} \big)^{\frac{n-2}{2}} -
\big(
\tfrac{n-2}{n} \big)^{\frac{n}{2}} \big)^{\frac{1}{n}}
= \big( \tfrac{n-2}{n} \big)^{\frac{1}{2}} \big( \tfrac{2}{n-2}
\big)^{\frac{1}{n}},
\end{equation*}
where $r_{\crit}$ is known as the {\em critical radius}. Note that
$r_{\crit}\approx 0.7274$ for
$n=3$, $r_{\crit}=1/\sqrt{2}\approx 0.7071$ for $n=4$, and $r_{\crit}$
is
strictly monotonically increasing and bounded from above by $1$ for
$n\geq 4$.

An important result of Khavinson and Neumann~\cite{KhavinsonNeumann:2006} states
that a rational harmonic function $f(z)-\conj{z}$, with $f(z)$ rational and of
degree $n$, may have at most $5(n-1)$
complex zeros (see also~\cite{KhavinsonNeumann2008}). Consequently, any
$n$-point lens may have at most $5(n-1)$ images. The lens modeled
by $R(z)$ from \eqref{eqn:unpert_rat_fun} and with $r<r_{\crit}$ attains this
upper bound for the
maximal number of images only for $n=3$.

Rhie~\cite{Rhie:2003} considered perturbations of $R(z)$ of the form
\begin{equation}
    \label{eqn:rhie_perturb}
    R_\eps(z) = (1-\eps) R(z) + \tfrac{\eps}{z} =
    \tfrac{ z^n - \eps r^n }{ z (z^n - r^n) },
\end{equation}
and indicated that for $n \ge 3$, $0 < \eps < 1$ ``small enough'', and
\begin{equation}\label{eqn:rhie_val}
    r  \defby (n-1)^{-\frac{1}{n}} \big( \tfrac{n-1}{n} \big)^{\frac{1}{2}} <
r_{\crit}
\end{equation}
the equation $R_\eps(z) = \conj{z}$ has $5n$ solutions. Since $R_\eps(z)$ is of
degree $n+1$,
Rhie's result implies the existence of maximal point lenses for any $n \geq 3$.

In section~\ref{sec:rhie} we reconsider Rhie's original construction
based on the function~\eqref{eqn:rhie_perturb}, and we extend Rhie's
result from the special value $r$ in \eqref{eqn:rhie_val} to all $r
\in (0,r_{\crit})$.  Further we discuss bounds on the eligible weight $\eps$ such
that the equation $R_\eps(z) = \conj{z}$ describes a maximal point lens.  In
particular, we derive a sharp bound.

In a related work, Bayer and Dyer~\cite{BayerDyer2007,BayerDyer2009}
intended to close some gaps in the understanding of maximal lensing
that remained after~\cite{Rhie:2003}. For $n\geq 3$ and $r\in
(0,1/\sqrt{2})$ they derived an upper bound $m_\ast>0$, so that for
any $\eps\in (0,m_\ast)$ the equation
\begin{equation}
\label{eqn:bdfun}
R(z)+\tfrac{\eps}{z}=\conj{z},
\end{equation}
has $5n$ solutions and thus models a maximal lensing case.

In section~\ref{sec:bd} we extend their result to a range of radii
that is strictly larger than the interval $(0, r_{\crit})$, and we give
sharp bounds for the mass $\eps$.
In the maximal lensing case for the function~\eqref{eqn:rhie_perturb}, the
maximal eligible mass $\eps$ can be expressed as a function of the radius $r$.
This is no longer the case for the function on the left hand side of
equation~\eqref{eqn:bdfun}, where these two parameters
show a more complex interdependence.
Despite this interdependence, we present a parameter transformation that
allows for an exact characterization of the maximal lensing case for the latter function.

\section{A concise proof of Rhie's construction}
\label{sec:rhie}

The following result was already shown by Mao, Petters, and Witt~\cite{MaoPettersWitt1999}.
We include a proof for completeness, and because a part of this proof will be used to prove
one of our main results below.

\begin{prop} \label{prop:num_sol}
Let $n \geq 3$ be given.  Depending on the parameter $r$, the
solutions to the equation $R(z) = \conj{z}$ with $R(z)$
from~\eqref{eqn:unpert_rat_fun} are given as follows:
\begin{compactitem}
\item If $0<r<r_{\crit}$, the equation has exactly $3n+1$ solutions
given by $z = 0$ and $3n$ values of $z$ of the form $r_1 e^{i
\frac{(2k+1) \pi}{n}}$, $r_2 e^{i \frac{(2k+1) \pi}{n}}$ and $r_3
e^{i \frac{2 k \pi}{n}}$, $k=0,1,\dots,n-1$, where $0 < r_1 < r_2 < 1
< r_3$.

\item If $r = r_{\crit}$, the equation has exactly $2n+1$
solutions given by $z = 0$, $( \tfrac{n-2}{n} )^{\frac{1}{2}} e^{
i \frac{(2k+1) \pi}{n}}$, and $r_3 e^{i \frac{2k \pi}{n}}$,
$k=0,1,\dots,n-1$, where $r_3>1$.

\item If $r > r_{\crit}$, the equation has exactly $n+1$ solutions
given by $z = 0$ and $r_3 e^{i \frac{2k \pi}{n}}$,
$k=0,1,\dots,n-1$, where $r_3>1$.
\end{compactitem}
\end{prop}

\begin{proof}
Obviously, $z=0$ solves $R(z) = \conj{z}$, regardless of the value of $r$.  To derive necessary and
sufficient conditions for nonzero solutions we use polar coordinates and write $z = \rho e^{i \varphi}$
with $\rho > 0$ and $\varphi \in \R$. Then a small manipulation shows that $R(z) = \conj{z}$ can be
written as
\begin{equation}
\label{eqn:comp_rho_phi}
( \rho^2 - 1 ) \rho^{n-2} e^{i n \varphi} = r^n.
\end{equation}
Since $r>0$, it is necessary that $\rho \neq 1$ and that $e^{i n \varphi}$ is
real.
Thus we must have either $e^{i n \varphi} = 1$ and hence $\rho > 1$,
or $e^{i n \varphi} = -1$ and hence $0 < \rho < 1$.

First consider the case $e^{i n \varphi} = 1$.  Then $\varphi = \tfrac{2 k \pi}{n}$
for some $k\in\{0,1,\dots,n-1\}$, and~\eqref{eqn:comp_rho_phi} becomes
$f_{+}(\rho)\defby \rho^n - \rho^{n-2} - r^n = 0$.
By Descartes' rule of signs, the polynomial $f_{+}(\rho)$ has exactly one
positive root, which
we denote by
$r_3>1$. Hence for all $r>0$ there exist $n$ solutions $r_3 e^{i
\frac{2 k \pi}{n}}$,
$k=0,1,\dots,n-1$, of the equation $R(z)=\conj{z}$.

Next consider the case $e^{i n \varphi} = -1$. Then $\varphi =
\tfrac{(2k+1)\pi}{n}$ for some
$k\in\{0,1,\dots,n-1\}$, and~\eqref{eqn:comp_rho_phi} becomes
$f_-(\rho) \coloneq \rho^n - \rho^{n-2} + r^n=0$.
The polynomial $f_-(\rho)$ has either zero or two positive roots. The only positive root of the derivative
$f_-'(\rho) = \rho^{n-3} ( n \rho^2 - (n-2) )$
is given by $\big( \tfrac{n-2}{n} \big)^{\frac{1}{2}}\in (0,1)$. Since $f_-(0) = f_-(1) = r^n > 0$,
the function $f_-(\rho)$ has two distinct roots $r_1,r_2\in (0,1)$ if and only if
\begin{equation*}
f_- \big( \big( \tfrac{n-2}{n} \big)^{\frac{1}{2}} \big) < 0, \quad\mbox{or, equivalently,}
\quad
r < \big( \big( \tfrac{n-2}{n} \big)^{\frac{n-2}{2}} - \big( \tfrac{n-2}{n} \big)^{\frac{n}{2}} \big)^{\frac{1}{n}}
= r_{\crit}.
\end{equation*}
Thus, for $r \in (0,r_{\crit})$ there are two positive roots
$0<r_1<r_2<1$ of $f_-(\rho)$, giving
$2n$ solutions $r_1 e^{i \frac{(2k+1)\pi}{n}}$ and $r_2 e^{i \frac{(2k+1)\pi}{n}}$, $k=0,1,\dots,n-1$,
of $R(z) = \conj{z}$.
For $r = r_{\crit}$, $\rho = (\tfrac{n-2}{n})^{\frac{1}{2}}$ is a (double)
positive root of $f_-(\rho)$, giving $n$ solutions
$(\tfrac{n-2}{n})^{\frac{1}{2}} e^{i \frac{(2k+1) \pi}{n}}$, $k = 0, 1, \ldots,
n-1$.
For $r > r_{\crit}$, there are no additional solutions of $R(z) = \conj{z}$.
\eop
\end{proof}

With an analogous but somewhat more technical proof we now show that in Rhie's
construction~\eqref{eqn:rhie_perturb}, maximal lensing occurs not only for $r$
from~\eqref{eqn:rhie_val}, but for \emph{all} $r \in (0,r_{\crit})$.

\begin{thm}
\label{thm:1}
Let $n \geq 3$ and $r \in (0, r_{\crit})$ be given. Denote by $\xi_1$
the smallest positive root of the polynomial $(n+2) \xi^n - n
\xi^{n-2} + 2 r^n$ and set
\begin{equation}
\label{eqn:eps_star}
\eps_\ast \defby \tfrac{\xi_1^{n+2} - \xi_1^n +
r^n \xi_1^2}{r^n}.
\end{equation}
Then $\eps_\ast \in (0,1)$, and for any $\eps\in (0,\eps_\ast)$,
the equation $\rpert(z) = \conj{z}$ has
$5n$ solutions of the form $r_2 e^{i \frac{(2k+1)\pi}{n}}$, $r_3 e^{i
\frac{(2k+1)\pi}{n}}$, $r_4 e^{i \frac{(2k+1)\pi}{n}}$, and $r_1 e^{i
\frac{2k\pi}{n}}$, $r_5 e^{i \frac{2k\pi}{n}}$, $k=0,1,\dots,n-1$,
where $0 < r_1 < \sqrt{\eps} < r_2 < r_3 < r_4 < 1 < r_5$.
If $\eps = \eps_*$ there exist only $4n$ solutions, and only $3n$ solutions
exist for $\eps > \eps_\ast$.
\end{thm}

\begin{proof}
Clearly, $z=0$ does not solve $\rpert(z) = \conj{z}$. As in the proof of
Proposition~\ref{prop:num_sol}, we now write $z = \rho e^{i \varphi}$,
where $\rho > 0$ and $\varphi \in \R$. Then $\rpert(z) = \conj{z}$ can be
written as
\begin{equation}\label{eqn:comp_rho_phi_pert}
(1-\rho^2) \rho^n e^{i n \varphi} = (\eps - \rho^2) r^n.
\end{equation}
Since the right hand side of \eqref{eqn:comp_rho_phi_pert} is real, $r>0$ (and
$\eps\neq 1$), it is
necessary that $\rho\neq1$, $\rho\neq \sqrt{\eps}$, and that either $e^{i n \varphi} = 1$ or $e^{i n \varphi} = -1$.

If $e^{i n \varphi} = 1$, then $\varphi = \tfrac{2k\pi}{n}$ for some
$k\in\{0,1,\dots,n-1\}$. Furthermore, $1-\rho^2$ and $\eps-\rho^2$ must have same sign, hence
either $0 < \rho < \sqrt{\eps}$ or $1 < \rho$, and \eqref{eqn:comp_rho_phi_pert} can be
written as
$f_+(\rho) \coloneq \rho^{n+2} - \rho^n - r^n \rho^2 + \eps r^n = 0$.
By Descartes' rule of signs $f_+(\rho)$ has either zero or two positive roots.
From
\begin{equation*}
f_+(0) = \eps r^n > 0, \quad \quad f_+(\sqrt{\eps}) = \eps^{\frac{1}{n}} (\eps-1) < 0,\quad
f_+(1) = (\eps-1) r^n < 0,
\end{equation*}
and $f_+(\rho)\rightarrow\infty$ for $\rho\rightarrow\infty$, we see that $f_+(\rho)$
indeed has one root $r_1\in (0, \sqrt{\eps})$ and one root $r_5\in (1, \infty)$.
Consequently, for all $r>0$ there exist $2n$ solutions
$r_1 e^{i \frac{2 k \pi}{n}}$ and $r_5 e^{i \frac{2 k \pi}{n}}$, $k=0,1,\dots,n-1$,
of the equation $R_\eps(z)=\conj{z}$.

If $e^{i n \varphi} = -1$, then $\varphi = \tfrac{(2k+1)\pi}{n}$ for some
$k\in\{0,1,\dots,n-1\}$. Here $1-\rho^2$ and $\eps-\rho^2$ must have opposite signs. Hence
$\sqrt{\eps} < \rho < 1$ is necessary, and \eqref{eqn:comp_rho_phi_pert} can be written as
$f_-(\rho) \coloneq \rho^{n+2} - \rho^n + r^n \rho^2 - \eps r^n = 0.$
The polynomial $f_-(\rho)$ has either one or three positive roots.
We will derive necessary and sufficient conditions so that $f_-(\rho)$ has three distinct positive
roots in the interval $(\sqrt{\eps},1)$. We start by noting that the positive roots of the derivative
$f_-'(\rho) = \rho \big( (n+2) \rho^n - n \rho^{n-2} + 2 r^n \big)$
are equal to the positive roots of
$g(\rho) \coloneq (n+2) \rho^n - n \rho^{n-2} + 2 r^n.$
From
\begin{equation*}
g'(\rho) =
n \rho^{n-3} \big( (n+2) \rho^2 - (n-2) \big),
\end{equation*}
we see that the unique positive root of $g'(\rho)$ is
$(\frac{n-2}{n+2})^{\frac{1}{2}} \in (0,1)$.
For all $r\in (0,r_{\crit})$ we have $g((\frac{n-2}{n+2})^{\frac{1}{2}}) < 0$.
Together with $g(0) = 2 r^n > 0$ and $g(1) = 2 + 2 r^n > 0$ this shows that $g(\rho)$
and thus $f_-'(\rho)$ have exactly two positive roots, say $\xi_1$ and $\xi_2$,
with $0 < \xi_1 < (\frac{n-2}{n+2})^{\frac{1}{2}} < \xi_2 < 1$.
Note that $f_-'(\rho)$ does not depend on $\eps$, so $\xi_1$ and $\xi_2$ are independent of $\eps$.

Let us write
\begin{equation*}
f_-(\rho) = \rho^2 p(\rho) - \eps r^n, \quad \text{ where } \quad p(\rho) \coloneq \rho^n - \rho^{n-2} + r^n.
\end{equation*}
From the proof of Proposition~\ref{prop:num_sol} we know that for all $r\in
(0,r_{\crit})$ the polynomial
$p(\rho)$ has exactly two distinct positive roots, say $z_1$ and $z_2$, where $0 < z_1 < z_2 < 1$.
Since $f_-(0) = f_-(z_1) = f_-(z_2)=- \eps r^n$, the mean value theorem implies that the only two
roots of $f_-'(\rho)$ satisfy $0 < \xi_1 < z_1 < \xi_2 < z_2$.  From $f_-(z_2) < 0 < f_-(1)$ we then
see that $f_-(\rho)$ has exactly one root $r_4\in (z_2, 1)$.

Further, $f_-(0) = f_-(z_1) < 0$ implies that $f_-(\rho)$ has two more
(distinct) roots if and only if
\begin{equation*}
f_-(\xi_1) > 0, \quad\mbox{or, equivalently,} \quad \eps < \tfrac{ \xi_1^2
p(\xi_1) }{r^n} = \eps_\ast.
\end{equation*}
Note that $\eps_\ast > 0$ since $p(\rho)>0$ on $(0, z_1)$. Further, $p(\rho)$ is
decreasing on $(0, z_1)$,
so that $p(\xi_1) < p(0) = r^n$, and thus $\eps_\ast < 1$. In summary, for all
$r\in (0,r_{\crit})$ and $\eps \in (0, \eps_\ast)$,
the function $f_-(\rho)$ has two more distinct roots $r_2, r_3$ with $\sqrt{\eps} < r_2 < \xi_1 < r_3 < z_1$.
Hence, for all $r\in (0,r_{\crit})$ and $\eps \in (0, \eps_\ast)$ there exist
$3n$ solutions
$r_2 e^{ i \frac{(2k+1) \pi}{n}}$, $r_3 e^{ i \frac{(2k+1) \pi}{n}}$ and $r_4 e^{ i \frac{(2k+1) \pi}{n}}$,
$k=0,1,\dots,n-1$, of the equation $R_\eps(z) = \conj{z}$.

On the other hand, if $\eps = \eps_\ast$, $\xi_1$ is a (double) zero of
$f_-(\rho)$.  Then $\rpert(z) = \conj{z}$ has the $2n$ solutions $\xi_1 e^{i
\frac{(2k+1) \pi}{n}}$ and $r_4 e^{i \frac{(2k+1) \pi}{n}}$, $k = 0, 1, \ldots,
n-1$, in addition to the $2n$ solutions corresponding to $r_1$ and $r_5$.
Finally, if $\eps > \eps_\ast$, then only the $n$ solutions
$r_4 e^{ i \frac{(2k+1) \pi}{n}}$, $k=0,1,\dots,n-1$, of $\rpert(z) =
\conj{z}$ occur (in addition to the $2n$ solutions corresponding to $r_1$ and
$r_5$).
\eop
\end{proof}

\begin{figure}[p]
\begin{center}
    \includegraphics[width=.48\textwidth]{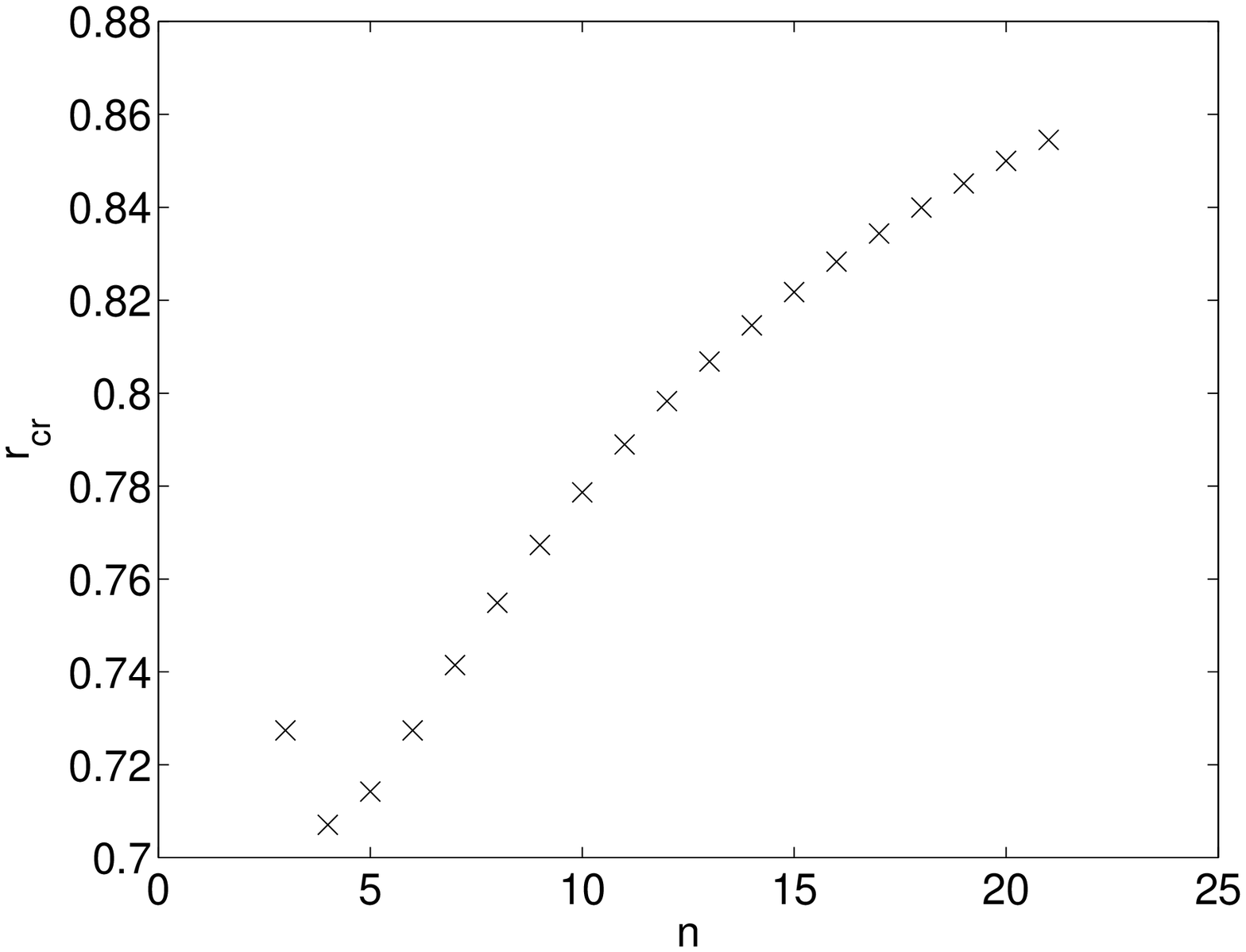}
    \includegraphics[width=.48\textwidth]{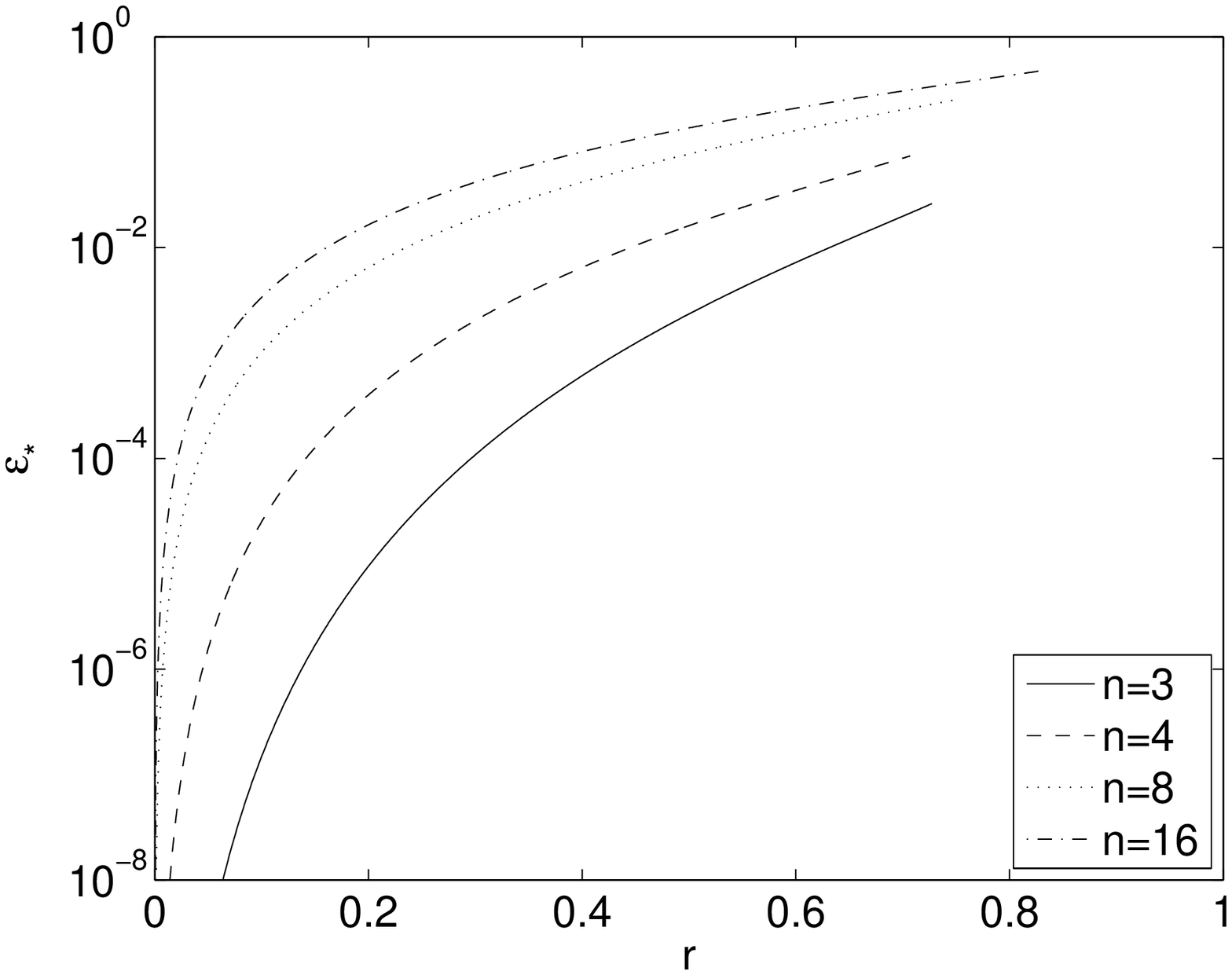}
\end{center}
    \caption{Left: The critical radius $r_{\crit}$ for $3\leq n\leq 21$.
    Right: $\eps_*$ from~\eqref{eqn:eps_star} for $n=3,4,8,16$ and $r\in (0,r_{\crit})$;
    $r_{\crit}(3) \approx 0.7274$,  $r_{\crit}(4) \approx 0.7071$, $r_{\crit}(8) \approx 0.7549$,
    $r_{\crit}(16) \approx 0.8283$. 
\label{fig:sharp_eps}}
\end{figure}

\begin{figure}[p]
    \includegraphics[width=.48\textwidth]{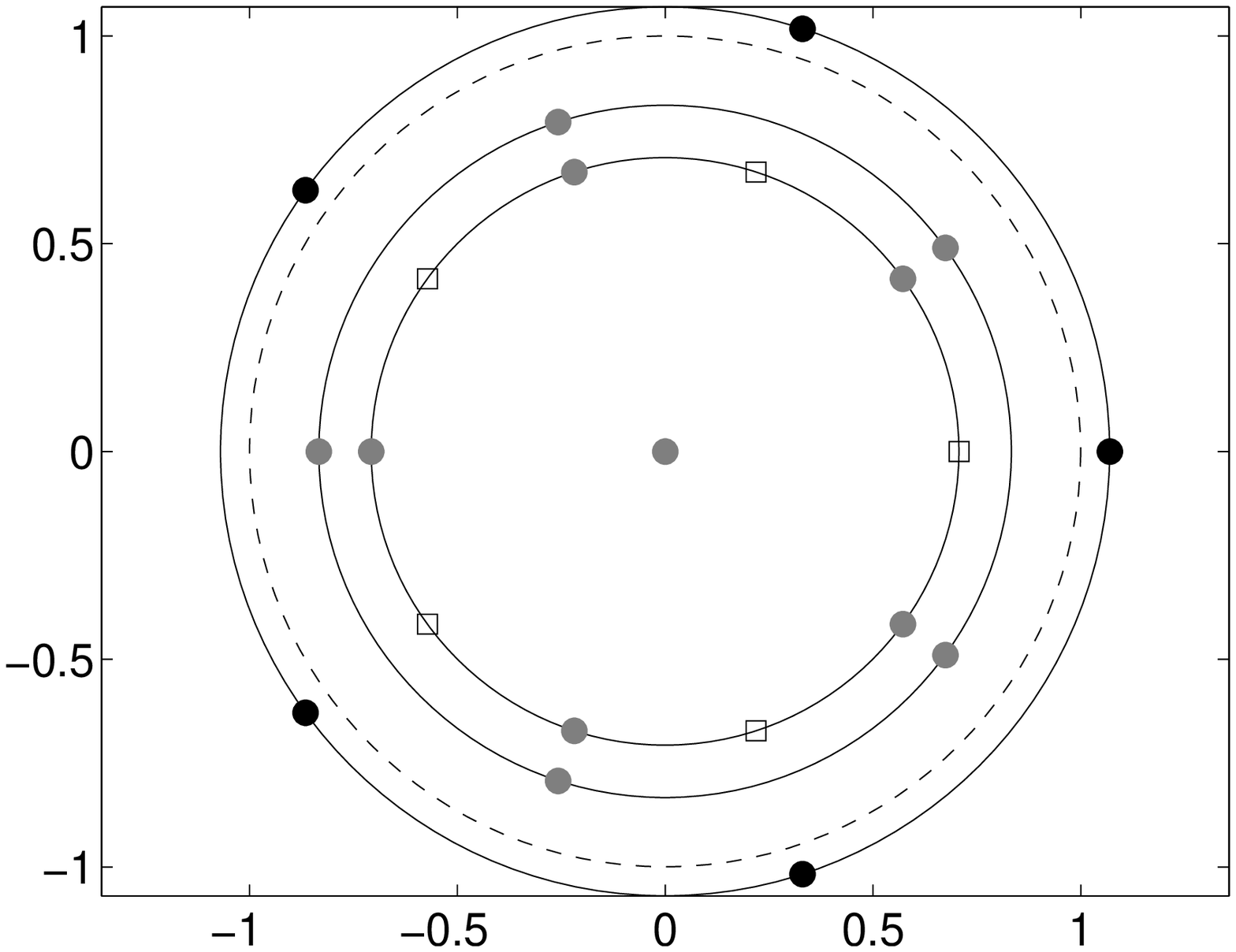}
    \hfill
    \includegraphics[width=.48\textwidth]{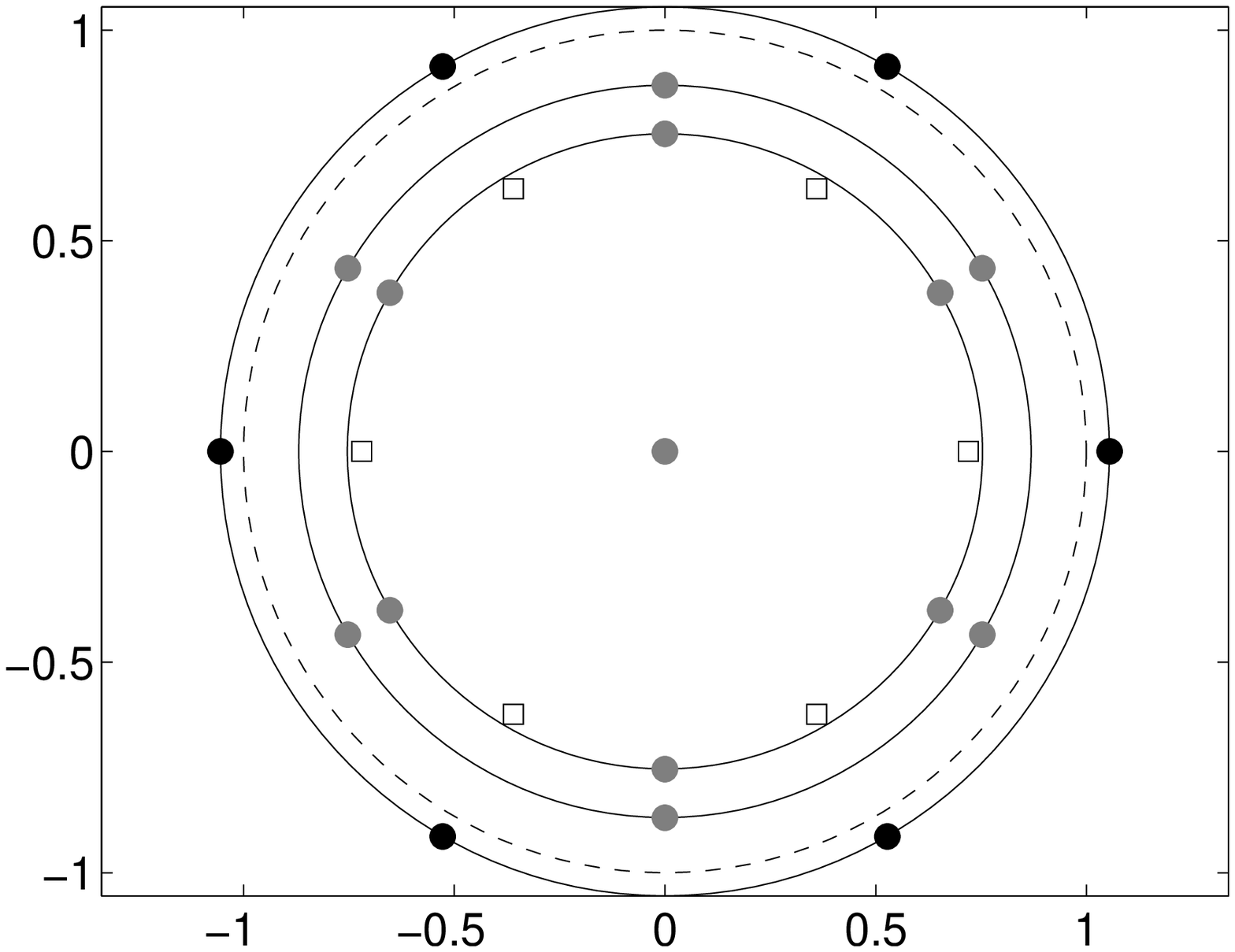}
\newline
    \includegraphics[width=.48\textwidth]{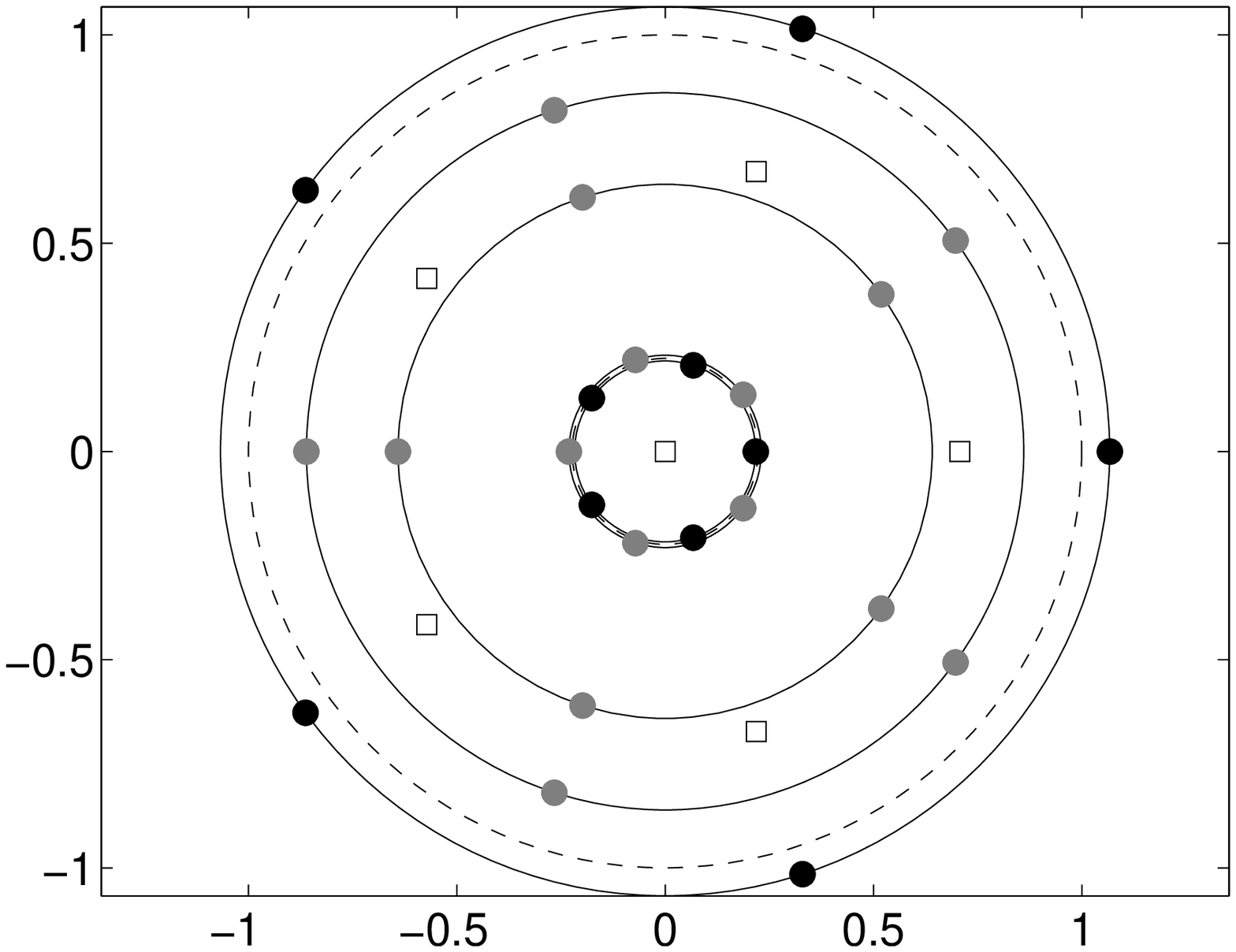}
    \hfill
    \includegraphics[width=.48\textwidth]{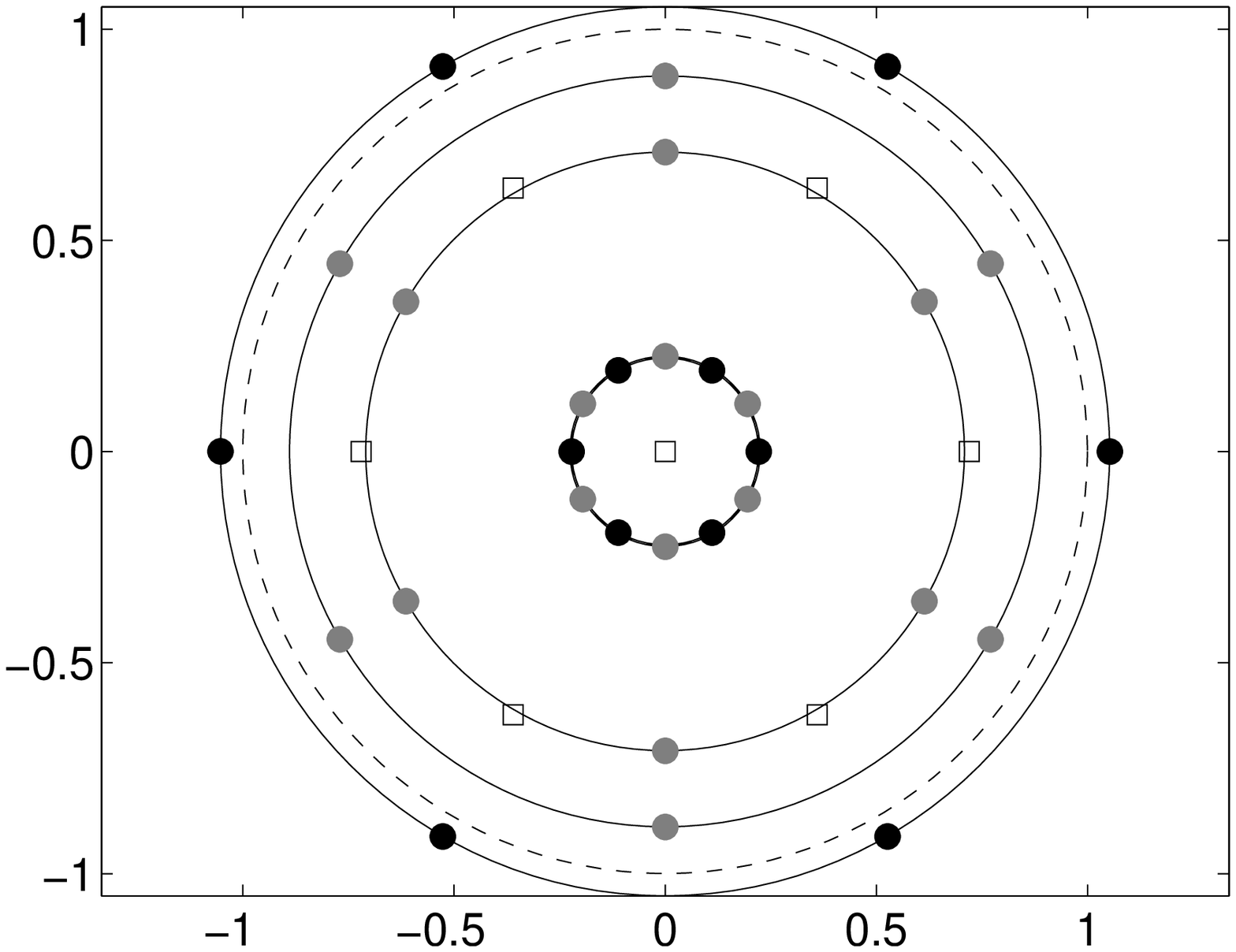}

\caption{Solutions of $R(z) = \conj{z}$
from~\eqref{eqn:unpert_rat_fun} (top) and $R_\eps(z) = \conj{z}$
from~\eqref{eqn:rhie_perturb} (bottom) for $n = 5$ (left) and $n = 6$
(right).  In both cases we used $\eps = 0.05$, and $r$ is chosen as
$0.99 \cdot r_{\crit}(n)$ (which is larger than the
radius~\eqref{eqn:rhie_val}).
The black points correspond to solutions with $e^{i n \varphi} = 1$
and the grey points are $z = 0$ (top) and the solutions with $e^{i n
\varphi} = -1$.  Squares indicate the mass points.  The dashed circles
are the unit circle and the circle of radius $\sqrt{\eps}$ (bottom
only).
\label{fig:Rhie_points}
}
\end{figure}

Note that $\eps_\ast$ in \eqref{eqn:eps_star} is a computable quantity
that depends on $n$ and $r$.
Fig.~\ref{fig:sharp_eps} displays $\eps_\ast$
for a few values of $n$ and $r\in (0,r_{\crit})$.
Proposition~\ref{prop:num_sol} and
Theorem~\ref{thm:1} are illustrated numerically in
Fig.~\ref{fig:Rhie_points}.
In particular,
the figure shows that the solutions of $R(z) = \conj{z}$ and
$\rpert(z) = \conj{z}$ are located on three and five circles around
the origin, respectively.

The bound on $\eps$ given in Theorem~\ref{thm:1} is sharp, but it is unclear
whether the value of $\eps_\ast$ can be parameterized explicitly
as a function of $n$ and $r$.  Next we
give an explicit parameterization of a slightly weaker bound on the
eligible mass $\eps$, depending only on $n$ and $r$, such that
$R_\eps(z) = \conj{z}$ admits the maximal number of solutions.

\begin{cor}
\label{cor:maxrhie_suff}
Let $n \ge 3$ and $r \in (0, r_{\crit})$ be given. The equation
$R_\eps(z) = \conj{z}$ has $5n$ solutions if
\begin{equation}
\label{eqn:bound_lls}
    0 < \eps < \tfrac{\zeta_0^{n+2} - \zeta_0^n + r^n \zeta_0^2}{r^n},
    \quad\mbox{where}\quad \zeta_0 \coloneq \tfrac{n+6}{n+8} r^{\frac{3n+1}{3n-6}}.
\end{equation}
\end{cor}
\begin{proof}
We use some notation from the proof of Theorem~\ref{thm:1}; we show
that $f_{-}(\zeta_0) > 0$ and that $r_2 < \zeta_0 < r_3$.  With some
algebraic manipulation one sees that the first assertion is equivalent
to~\eqref{eqn:bound_lls}.  Assuming that $f_{-}(\zeta_0) > 0$, the
second assertion is implied by $\zeta_0 <
(\frac{n-2}{n+2})^{\frac{1}{2}}$, which follows from a rather
technical calculation (Lemma~\ref{lem:zeta_eta}).
\eop
\end{proof}

\begin{figure}[t]
\begin{center}
    \includegraphics[width=.48\textwidth]{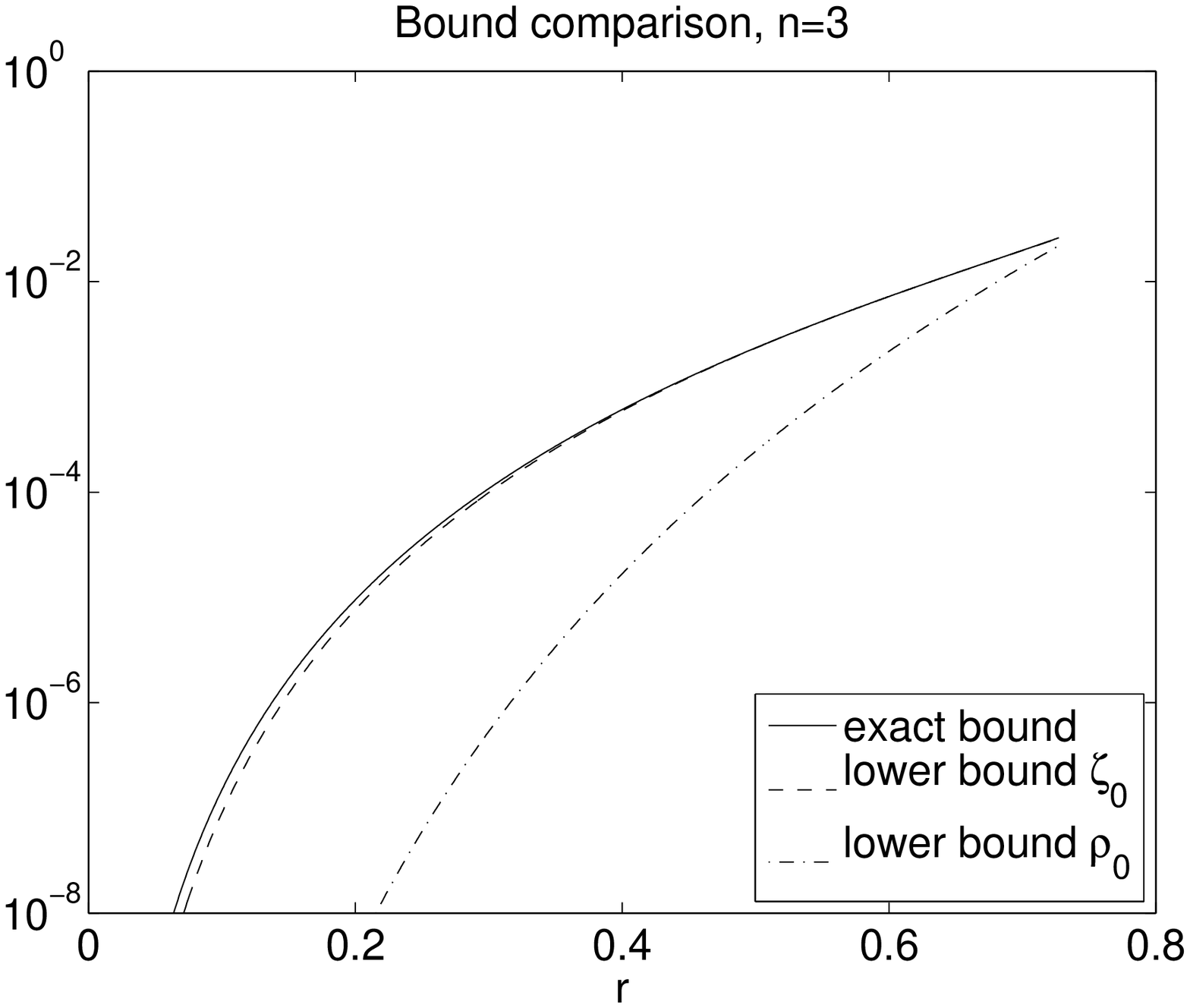}
    \hfill
    \includegraphics[width=.48\textwidth]{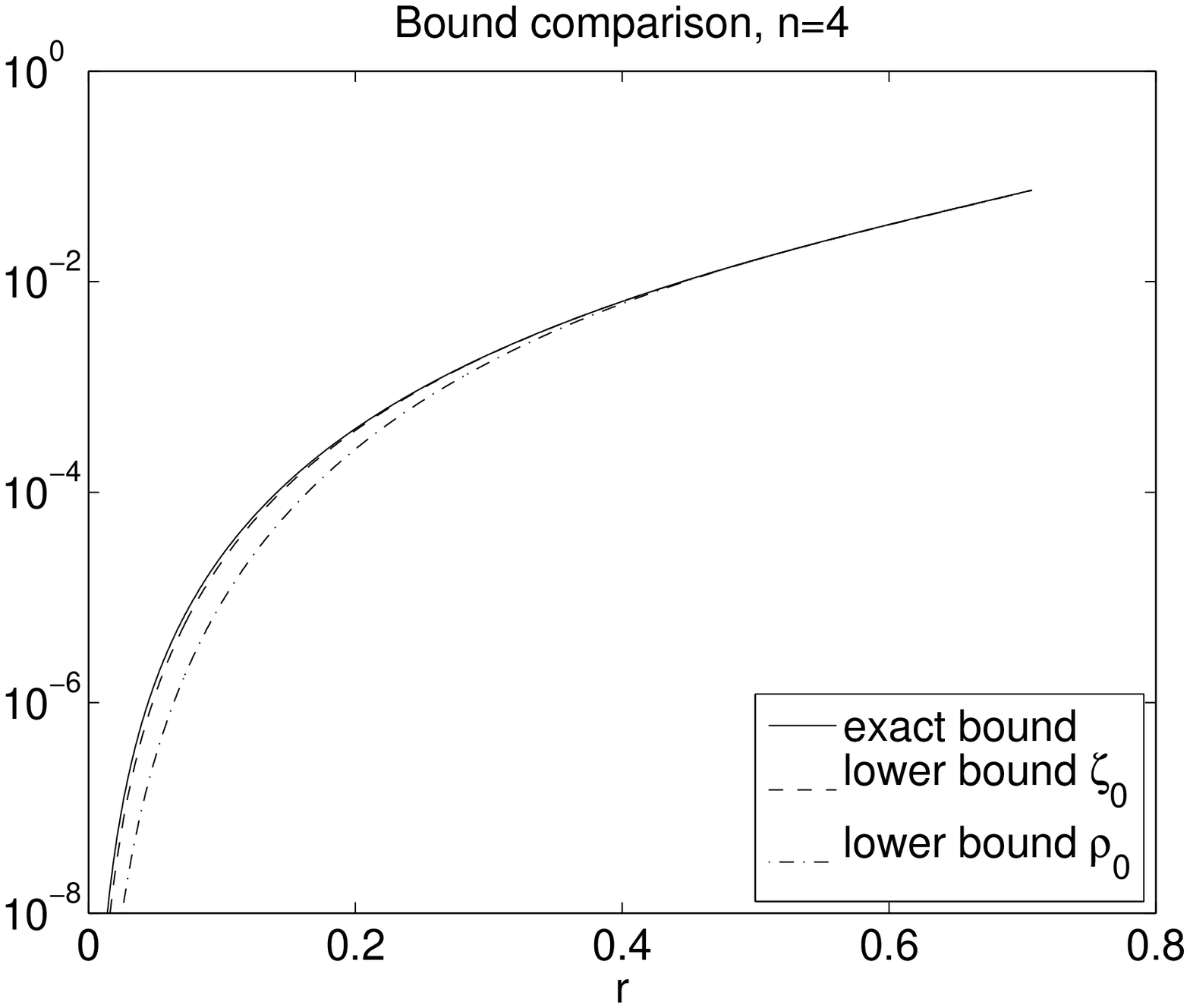}
    \newline
    \includegraphics[width=.48\textwidth]{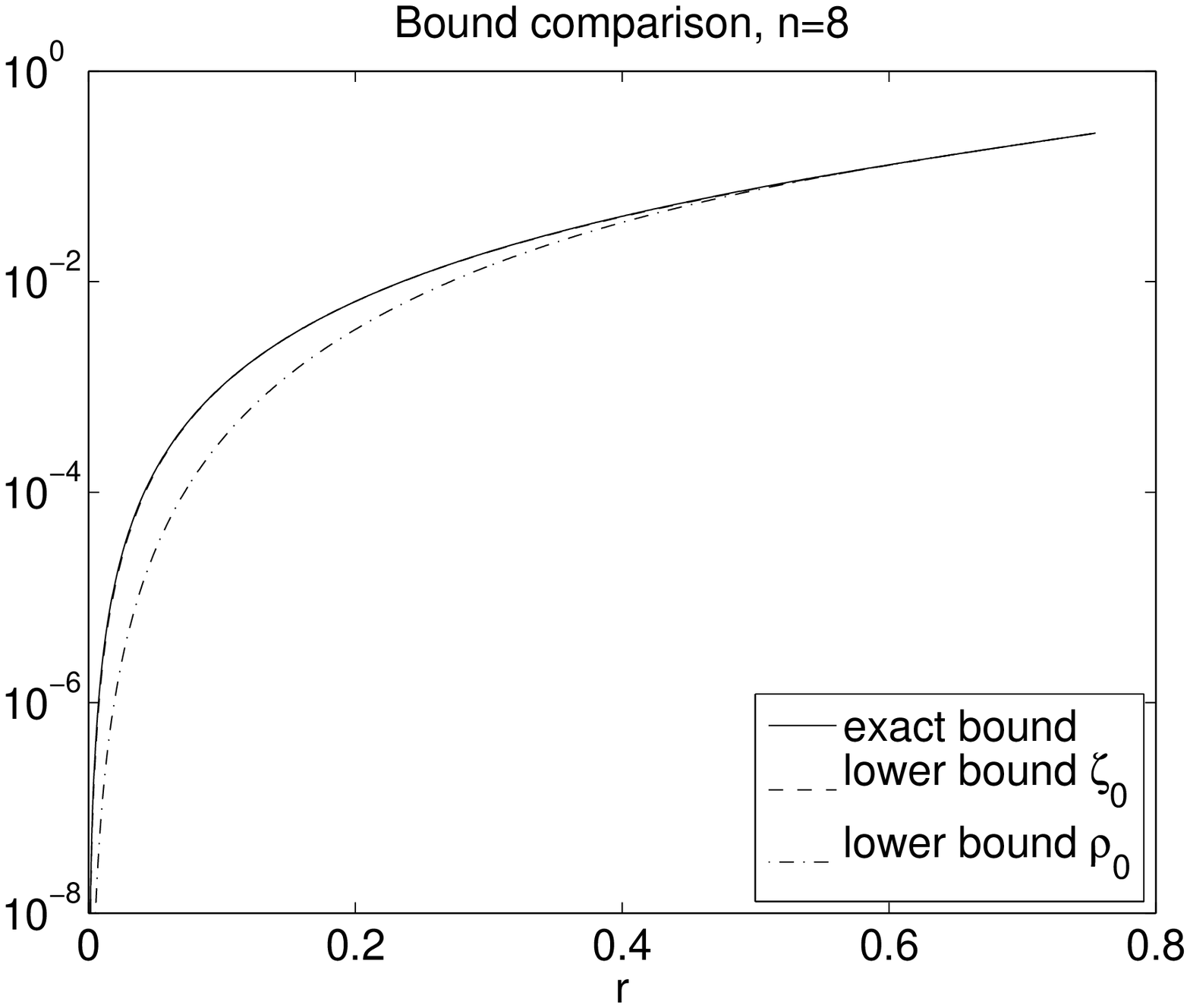}
    \hfill
    \includegraphics[width=.48\textwidth]{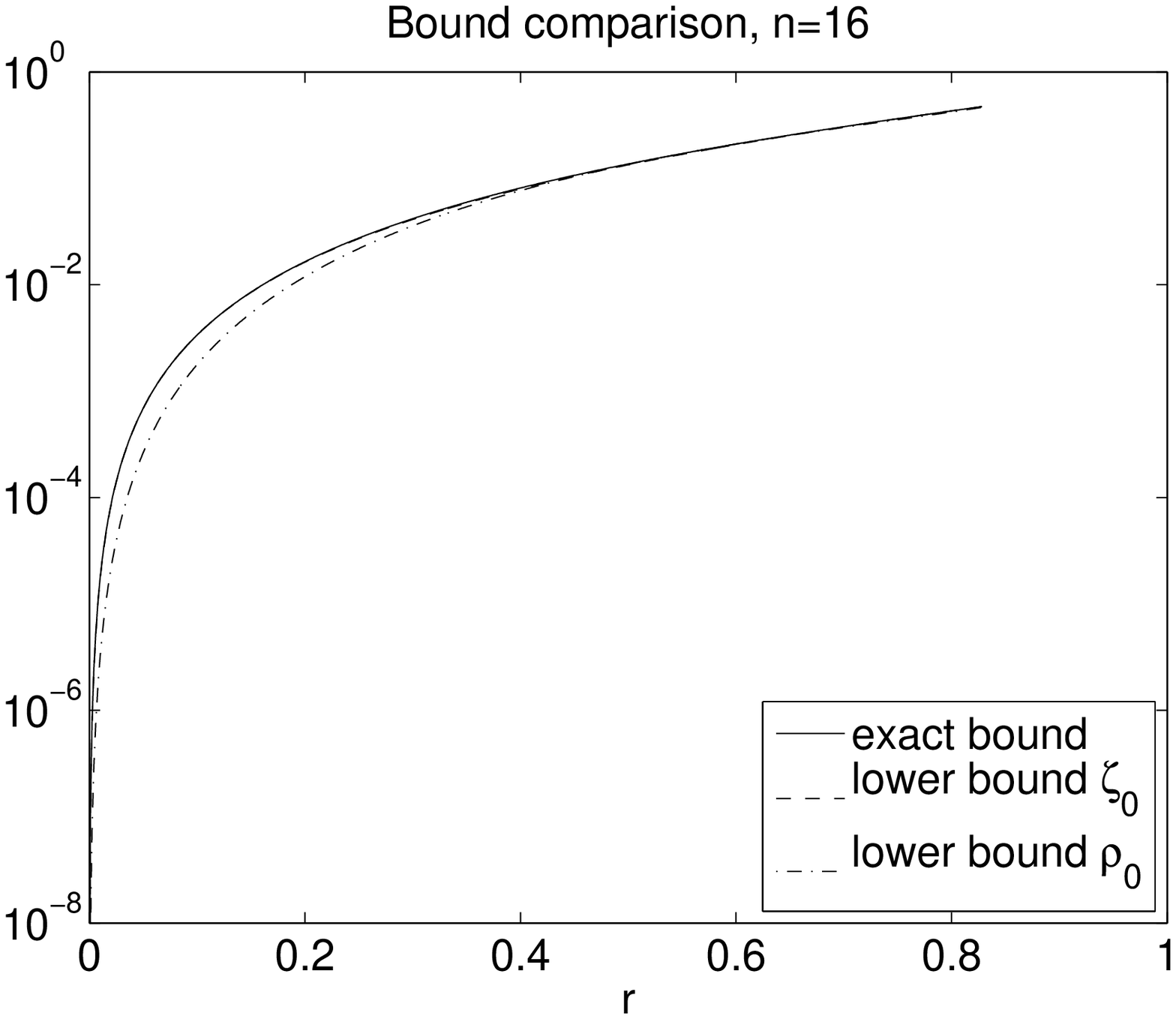}
\end{center}
    \caption{Comparison of mass bounds that imply maximal lensing.
For $n=8$ and $n=16$, the exact bound and the lower bound implied by
$\zeta_0$ are visually indistinguishable. Values below $10^{-8}$ have
been clipped.
\label{fig:rhie_bound_cmp}}
\end{figure}

\begin{rem}
\label{rem:rho0_rhie}
An alternative bound can be derived from the approach taken
in~\cite{BayerDyer2007}.  There, instead of using $\zeta_0$ from above
as a ``trial point'' at which $f_{-}(\rho)$ is positive, the authors
choose $\rho_0 = r^{1 + \frac{6}{n}}$.  Note that $\rho_0 <
(\tfrac{n-2}{n+2})^{\frac{1}{2}}$ holds for $n \geq 3$ and $r \in (0,
r_{\crit})$.  Applied to our context, the resulting mass bound becomes
\begin{equation}
    \label{eqn:bound_bd}
    0 < \eps < \tfrac{\rho_0^{n+2} - \rho_0^n + r^n \rho_0^2}{r^n} =
    r^{\frac{12}{n} + 8} - r^6 + r^{\frac{12}{n} + 2}.
\end{equation}
\end{rem}

The mass constraint~\eqref{eqn:bound_lls} is displayed in
Fig.~\ref{fig:rhie_bound_cmp} (``lower bound $\zeta_0$'') along with the
exact bound given in Theorem~\ref{thm:1} and the
bound from~\eqref{eqn:bound_bd} (``lower bound $\rho_0$'').

In Theorem~\ref{thm:1} the bound $r_{\crit}$ on $r$ is suggested by
Proposition~\ref{prop:num_sol}.  Numerical experiments show that the equation
$R_\eps(z) = \conj{z}$ can still have $5n$ solutions for $r$ slightly
larger than $r_{\crit}$.  We did not try to quantify the exact upper
bound on eligible $r$ such that $R_\eps(z) = \conj{z}$ has $5n$ solutions.
However, we can show the following.

\begin{cor}
    \label{cor:r_ub_rhie}
For any $\eps \in (0,1)$ and $r \geq
(\tfrac{n-2}{n+2})^{\frac{n-2}{2n}}$, the equation $R_\eps(z) =
\conj{z}$ has exactly $3n$ solutions.
\end{cor}

\begin{proof}
In the notation of the proof of Theorem~\ref{thm:1}, we have
\begin{equation*}
g( (\tfrac{n-2}{n+2})^{\frac{1}{2}} )
= (\tfrac{n-2}{n+2})^{\frac{n-2}{2}} \big( (n+2) \tfrac{n-2}{n+2} - n \big) + 2
r^n
= -2 (\tfrac{n-2}{n+2})^{\frac{n-2}{2}} + 2 r^n.
\end{equation*}
Thus, $r \geq (\tfrac{n-2}{n+2})^{\frac{n-2}{2n}}$ is equivalent to $g((\tfrac{n-2}{n+2})^{\frac{1}{2}})
\geq 0$.  Since $g(\rho)$ has its minimum at this point, we have $g(\rho) \geq
0$ for $\rho > 0$ with equality at most for $\rho =
(\tfrac{n-2}{n+2})^{\frac{1}{2}}$. Now, since $f_-'(\rho) = \rho g(\rho)$,
the function $f_-(\rho)$ is strictly monotonically increasing and has only one
positive zero.  This implies that $R_\eps(z) = \conj{z}$ has only $3n$ solutions.
\eop
\end{proof}

\section{Extension to additive perturbations}
\label{sec:bd}

In Rhie's lens model~\eqref{eqn:rhie_perturb} the perturbation does not
alter the total mass of the lens. In
fact, since~\eqref{eqn:unpert_rat_fun} is normalized, the total mass
remains normalized to unity after the perturbation.

Another way to model a small mass perturbation
of~\eqref{eqn:unpert_rat_fun} is to \emph{increase} the total lens
mass by placing an \emph{additional} mass in the origin.  This kind of
perturbation was considered by Bayer and Dyer~\cite{BayerDyer2007}.
Mathematically, the resulting image configuration of the perturbed
lens corresponds to the solutions to the equation $S_{r,\eps}(w) =
\conj{w}$, where
\begin{equation}
    \label{eqn:bd_perturb}
    S_{r,\eps} (w) = R_r(w) + \tfrac{\eps}{w} =
    \tfrac{w^{n-1}}{w^n - r^n} + \tfrac{\eps}{w} =
    \tfrac{(1 + \eps) w^n - \eps r^n}{w (w^n - r^n)}.
\end{equation}
Here and in what follows, we write $R_r(z)$ and $R_{r, \eps}(z)$ for
the functions in~\eqref{eqn:unpert_rat_fun}
and~\eqref{eqn:rhie_perturb} in order to make the dependence on $r$ and $\eps$
explicit.

For notational clarity, we will always use $w$ as the variable in the
function~\eqref{eqn:bd_perturb}, as opposed to $z$ which we will
reserve for~\eqref{eqn:rhie_perturb}.  In particular, we will show how
to transform solutions from one equation to the other, which
necessitates the use of different variable names.

\begin{lem}
\label{lem:scaling}
Let $T(z) = \sum_{j=1}^{n} \frac{\sigma_j}{z - z_j}$ be a rational
function and $c > 0$.  Then
\begin{equation*}
    c T(z) = \conj{z} \;\;\Leftrightarrow\;\;
    \widetilde{T}(w) = \conj{w},
\end{equation*}
where $\widetilde{T}(w) = \sum_{j=1}^{n} \frac{\sigma_j}{w - w_j}$, $w_j =
\frac{z_j}{\sqrt{c}}$ and $w = \frac{z}{\sqrt{c}}$.
\end{lem}
\begin{proof}
From
\begin{equation*}
c T(z) = \sqrt{c} \sum_{j=1}^n \frac{\sigma_j}{\frac{z}{\sqrt{c}} -
\frac{z_j}{\sqrt{c}}}
\end{equation*}
we easily see that $c T(z) = \conj{z}$ if and only if $\sum_{j=1}^n \frac{\sigma_j}{w - w_j}  = \conj{w}$.
\eop
\end{proof}

The lemma can be used to map solutions of $R_{r, \eps}(z) =
\conj{z}$ and $S_{s, \delta} (w) = \conj{w}$ back and forth, as we
will show next.
\begin{prop}
\label{prop:rhie_bd_trafo}
Let $n \ge 3$, $r \in (0, r_{\crit})$ and $\eps\in(0,1)$ be given.
Then
\begin{equation*}
    R_{r, \eps}(z) = \conj{z} \;\;\Leftrightarrow\;\; S_{s, \delta}(w) = \conj{w},
\end{equation*}
where $s = \frac{r}{\sqrt{1 - \eps}}$, $\delta = \frac{\eps}{1 - \eps}$
and $w = \frac{z}{\sqrt{1 - \eps}}$.
\end{prop}
\begin{proof}
Since
\begin{equation*}
R_{r, \eps}(z) = (1 - \eps) R_r(z) + \tfrac{\eps}{z}
= (1 - \eps) (R_r(z) + \tfrac{\delta}{z}),
\end{equation*}
Lemma~\ref{lem:scaling} shows that $R_{r, \eps}(z) = \conj{z}$ if and only if
$S_{s, \delta}(w) = R_s(w) + \tfrac{\delta}{w} = \conj{w}$.
\eop
\end{proof}

The correspondence of the parameters $r, s$ and $\eps, \delta$ can be
written conveniently as the transformation
\begin{equation}
\label{eqn:phi_trafo}
(s, \delta) = \Phi(r,\eps) = ( \tfrac{r}{\sqrt{1 - \eps}},
\tfrac{\eps}{1 - \eps}).
\end{equation}
Note that $\Phi$ is bijective with inverse
$(r, \eps) = \Phi^{-1} ( s, \delta )  = ( \tfrac{s}{\sqrt{1 + \delta}},
\tfrac{\delta}{1 + \delta})$.
We will use the transformation $\Phi$ in order to express an analogue
to Theorem~\ref{thm:1} for~\eqref{eqn:bd_perturb}.

\begin{thm}
\label{thm:maxlens_bd}
Let $n \ge 3$ and $r \in (0, r_{\crit})$ be given.  The equation
$S_{s, \delta}(w) = \conj{w}$ has $5n$ solutions if and only if $\eps \in
(0,\eps_*)$, with $\eps_*$ from \eqref{eqn:eps_star}, and
where the parameters $r, \eps$ and $s, \delta$ are
connected via~\eqref{eqn:phi_trafo}.
These solutions are located on five circles with radii $0 < s_1 <
\sqrt{\delta} < s_2 < s_3 < s_4 < \sqrt{1 + \delta} < s_5$.
\end{thm}
\begin{proof}
From Proposition~\ref{prop:rhie_bd_trafo} we know that solutions to $S_{s,
\delta}(w) = \conj{w}$ can be mapped to solutions to $R_{r,
\eps}(z) = \conj{z}$ via
$w = \frac{z}{\sqrt{1 - \eps}} = z \sqrt{1 + \delta}$ (and vice
versa).  Using this relation, the assertions follow directly from
Theorem~\ref{thm:1}.
\eop
\end{proof}

\begin{figure}[t]
\begin{center}
    \includegraphics[width=.48\textwidth]{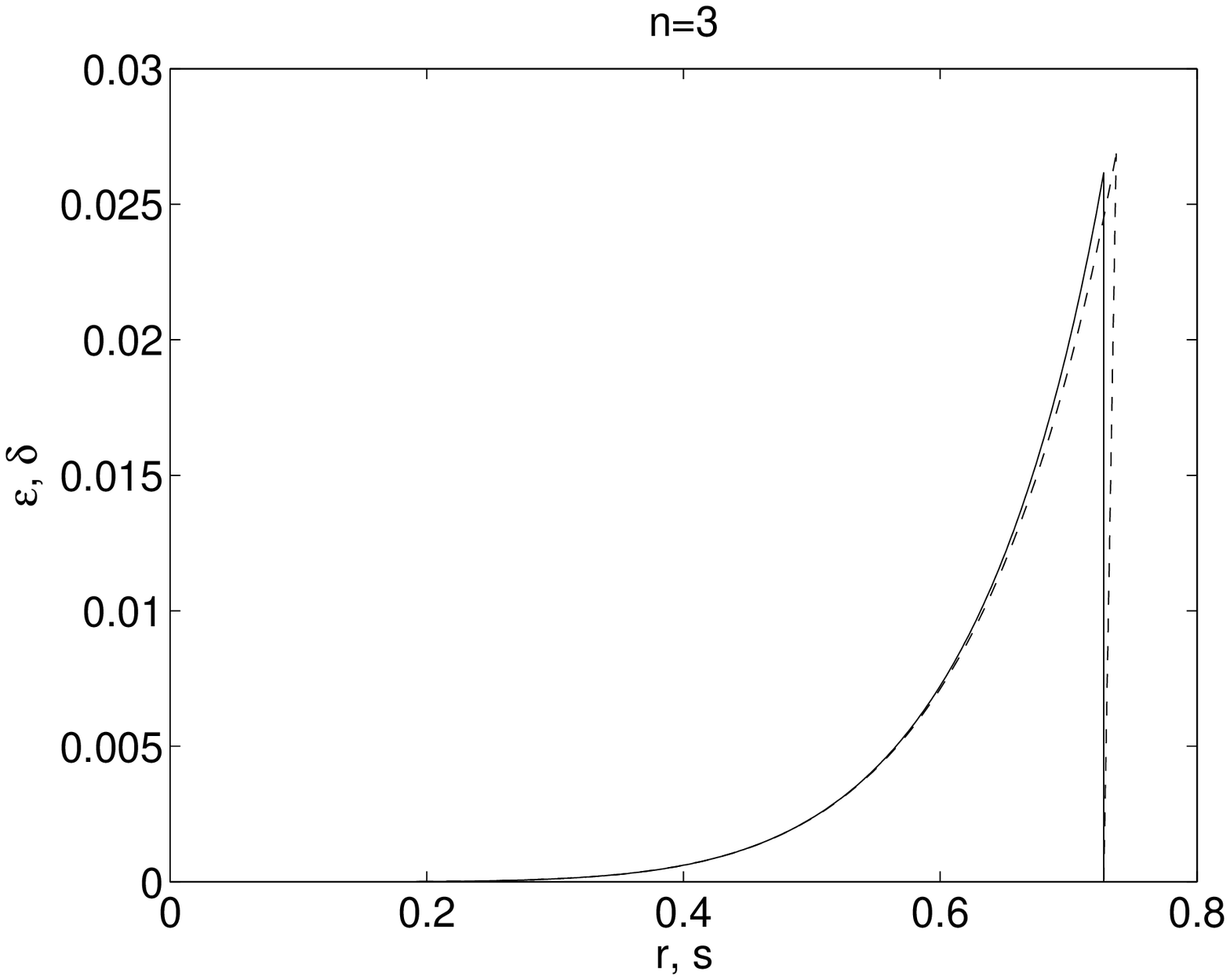}
    \hfill
    \includegraphics[width=.48\textwidth]{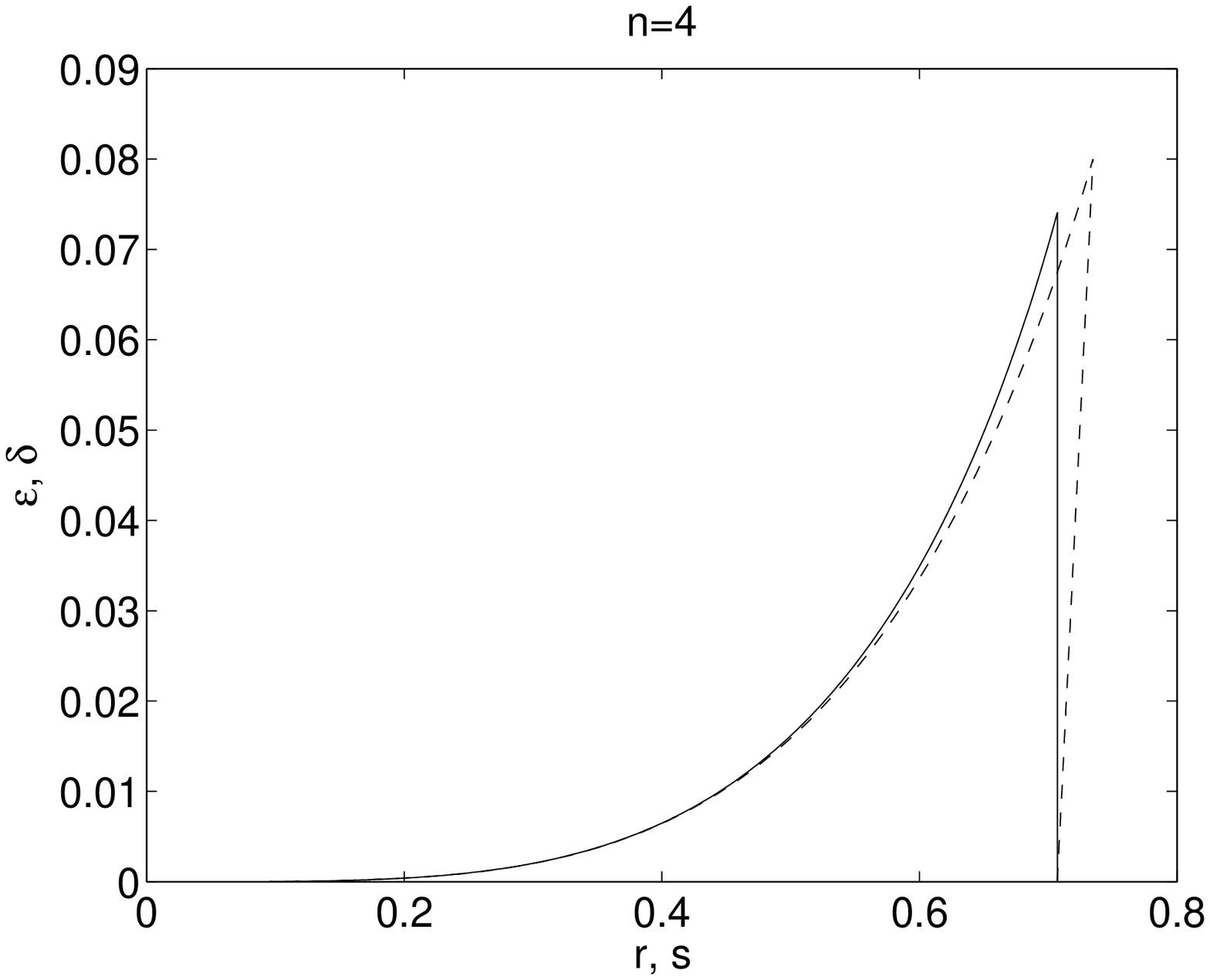}
    \newline
    \includegraphics[width=.48\textwidth]{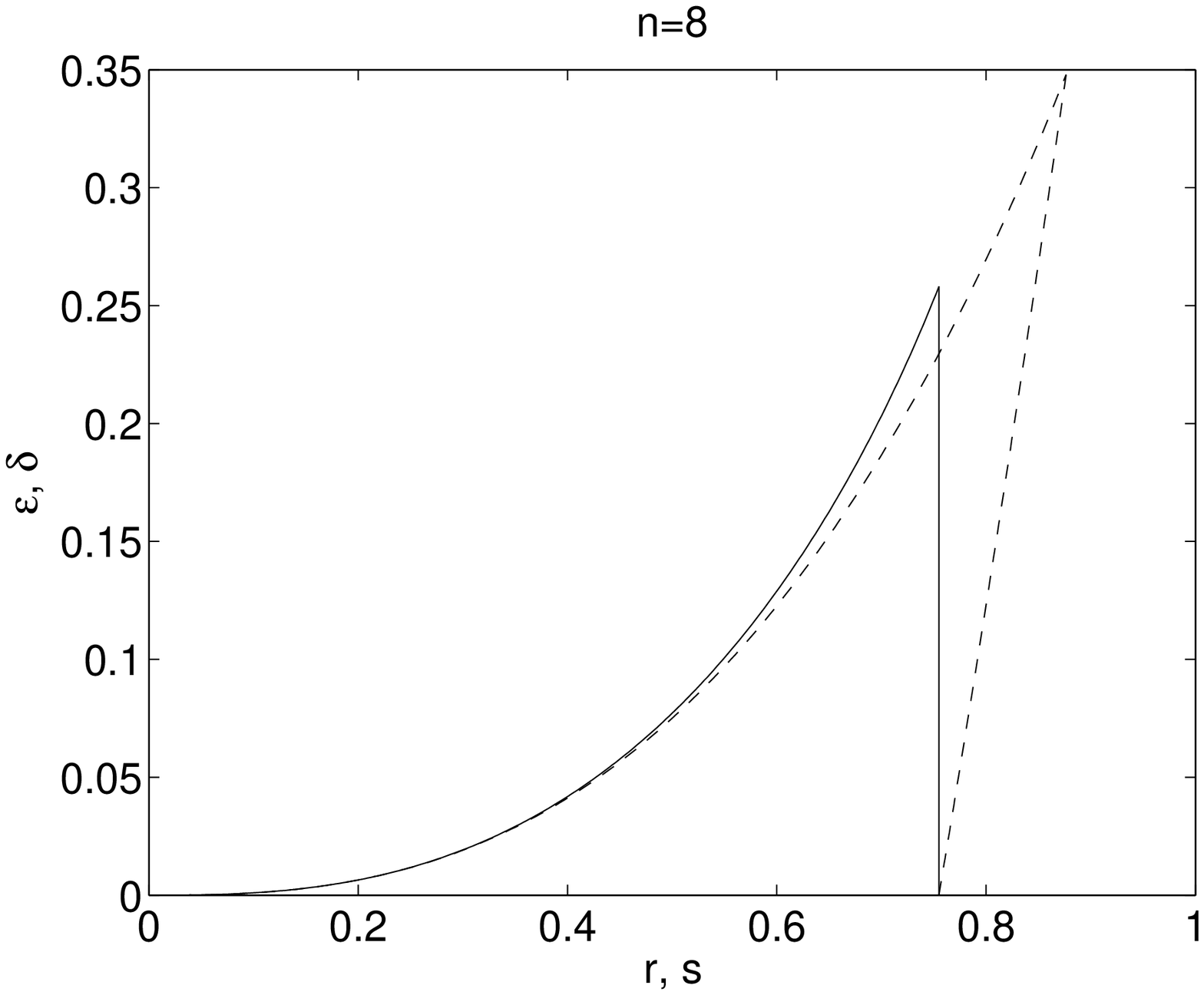}
    \hfill
    \includegraphics[width=.48\textwidth]{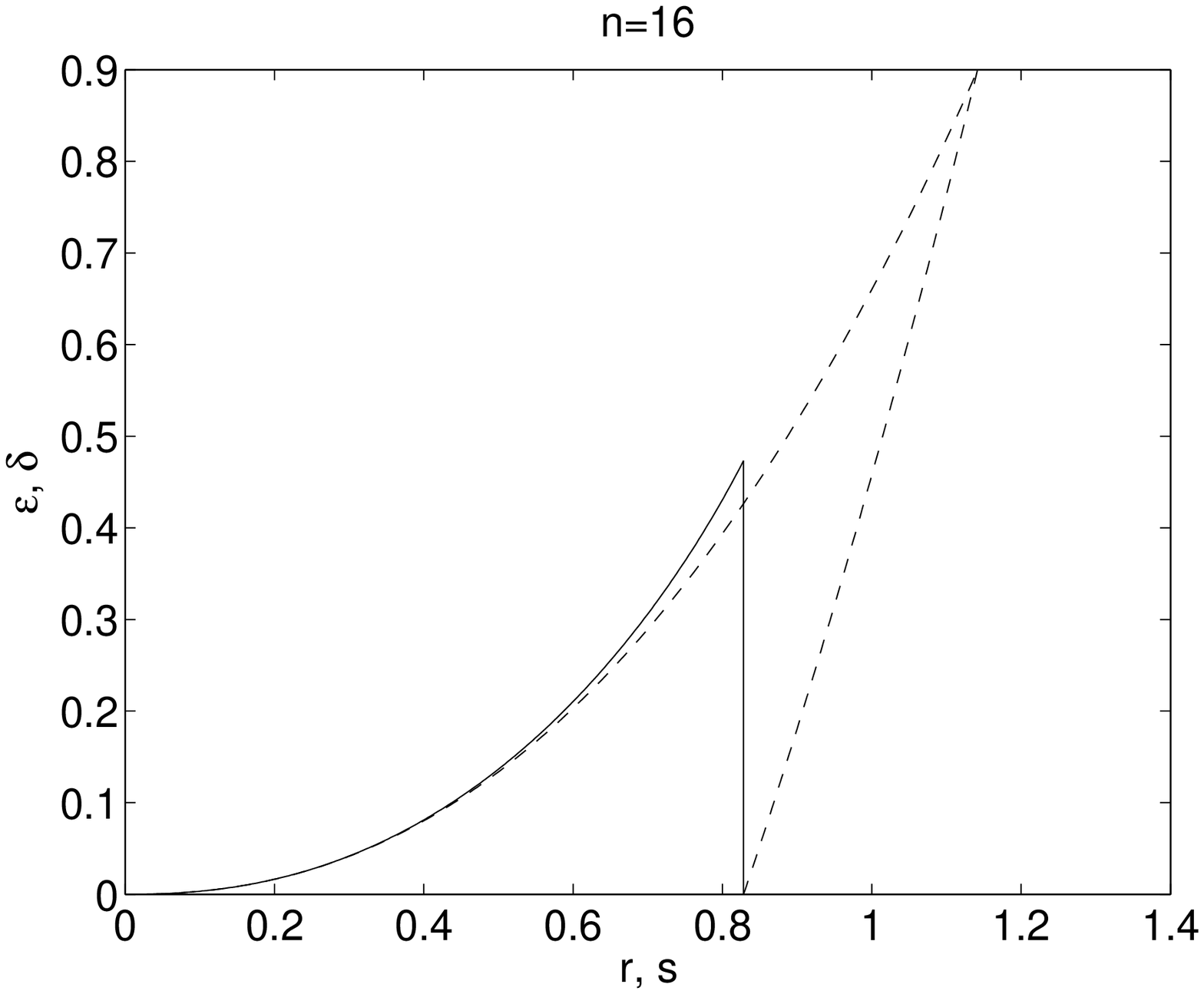}
\end{center}
    \caption{Comparison of parameters $(r, \eps)$ that induce maximal
lensing in $R_{r,\eps}(z) = \conj{z}$ with corresponding parameters
$(s, \delta)$ for $S_{s, \delta}(w) = \conj{w}$.  The region enclosed
by the horizontal axis and the solid lines contain all such admissible $(r,
\eps)$ pairs, while the transformed $(s, \delta)$ pairs are enclosed by
the horizontal axis and the dashed lines.
\label{fig:eps_delta}}
\end{figure}

Theorem~\ref{thm:maxlens_bd} describes parameter pairs $(s, \delta)$
such that maximal lensing occurs.  Note that $s$ is \emph{not} limited
by $r_{\crit}$, and that $s$ and $\delta$ are functionally dependent
via the transformation~\eqref{eqn:phi_trafo}.  In
Fig.~\ref{fig:eps_delta} we compare the range of admissible $(r,
\eps)$ pairs with $r < r_{\crit}$ with the transformed pairs $(s,
\delta)$.  Note also that the upper bound on $\delta$ shown in this
figure is \emph{sharp} for all values of $s$, i.e., maximal lensing
cannot happen for values of $\delta$ larger than the ones described in
Theorem~\ref{thm:maxlens_bd}.

In light of Corollary~\ref{cor:r_ub_rhie}, the plots suggest
that there is a \emph{lower} bound on the mass $\delta$, such that
maximal lensing happens for values $s > r_{\crit}$.  An example for
this phenomenon is given in Fig.~\ref{fig:bd_good_bad}.

\begin{figure}[t]
\begin{center}
    \includegraphics[width=.48\textwidth]{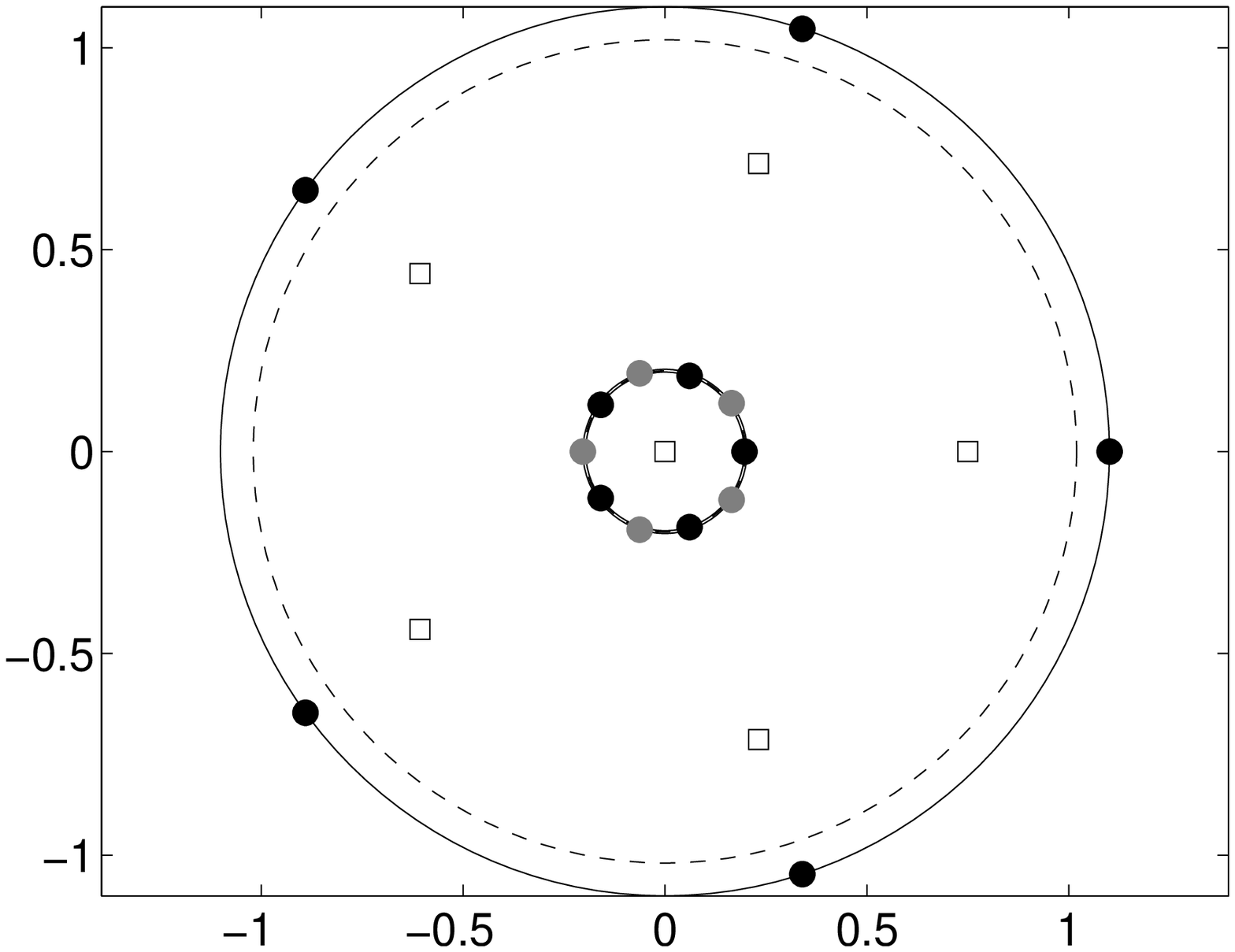}
    \hfill
    \includegraphics[width=.48\textwidth]{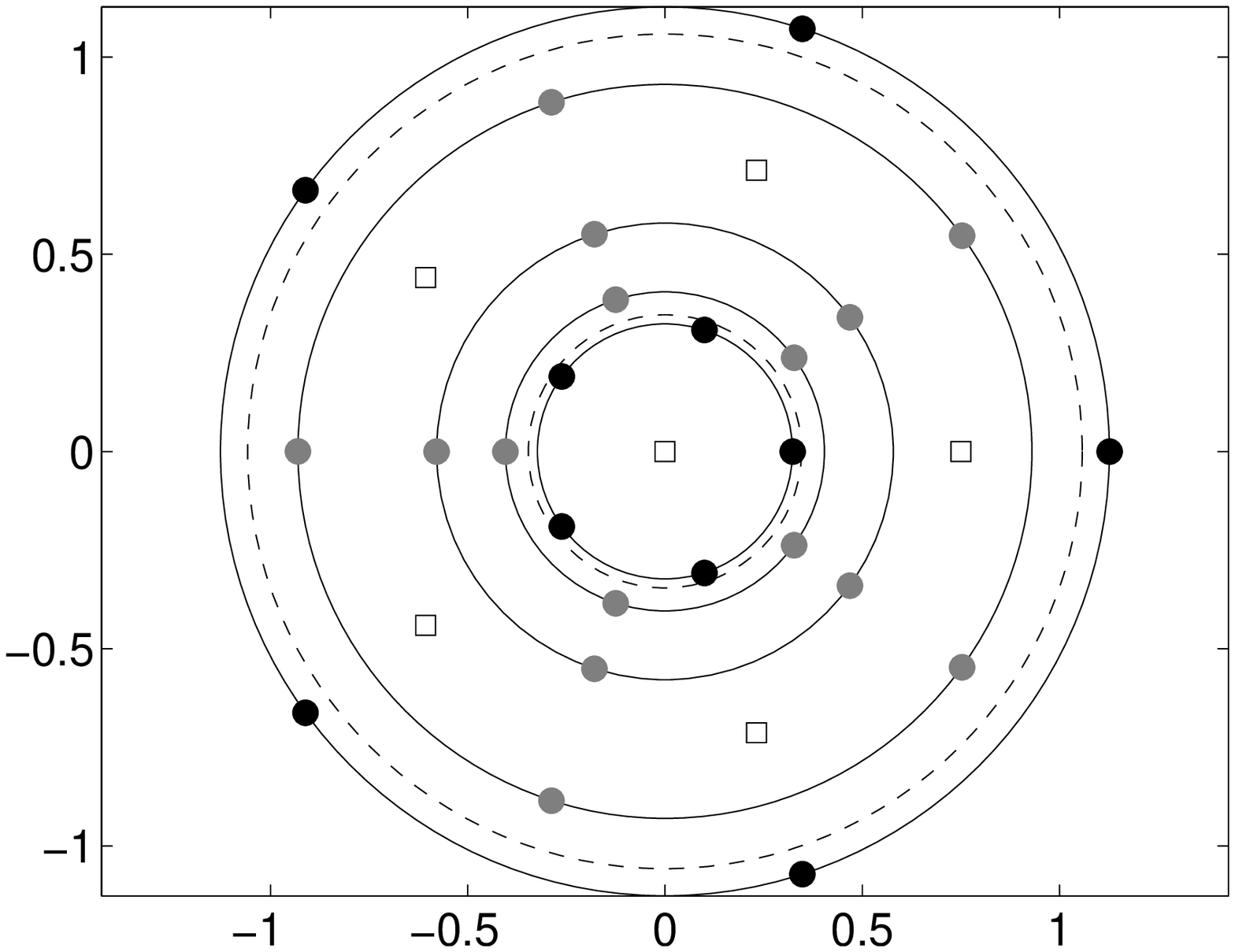}
\end{center}
    \caption{
Computed solutions to the equation $S_{s, \delta}(w) = \conj{w}$ with
$n=5$, $s=0.75 > r_{\crit}$ for masses $\delta_1 = 0.04$ (left) and
$\delta_2 = 0.12$ (right). The symbols are as in
Fig.~\ref{fig:Rhie_points}.  The mass $\delta_1$ is not sufficiently
large for this value of $s$ to induce maximal lensing.
\label{fig:bd_good_bad}}
\end{figure}

\begin{rem}
By repeating the proof of Theorem~\ref{thm:1} with respect
to~\eqref{eqn:bd_perturb},  another representation of the sharp upper
bound on $\delta$ can be obtained as
\begin{equation*}
    \delta <
     \tfrac{\xi_1^{n+2} - \xi_1^n + s^n \xi_1^2}{s^n + \xi_1^n},
\end{equation*}
where $\xi_1$ is the smallest positive root of
$(n+2) \xi^n - (1+\delta) n \xi^{n-2} + 2 s^n$.  Note that in contrast
to the situation in Theorem~\ref{thm:1}, $\xi_1$ now depends on
$\delta$, which makes this representation rather difficult to interpret.
For convenience, we repeat the proof for this case in
Theorem~\ref{thm:2}.
\end{rem}

We next present an explicit bound for the mass $\delta$, depending
only on $s$ and $n$, which implies maximal lensing.  The proof follows
from Theorem~\ref{thm:2} and by repeating the argument used in
Corollary~\ref{cor:maxrhie_suff}.
\begin{thm}
Let $n \ge 3$ and $s \in (0, r_{\crit})$ be given. The equation
$S_{s, \delta} (w) = \conj{w}$ has $5n$ solutions if
\begin{equation}
\label{eqn:bdmax_zeta}
    \delta < \tfrac{\zeta_0^{n+2} - \zeta_0^n + s^n \zeta_0^2}{s^n + \zeta_0^n},
    \quad\mbox{where}\quad \zeta_0
\defby \tfrac{n+6}{n+8} s^{\frac{3n+1}{3n-6}}.
\end{equation}
\end{thm}

\begin{figure}[t]
\begin{center}
    \includegraphics[width=.48\textwidth]{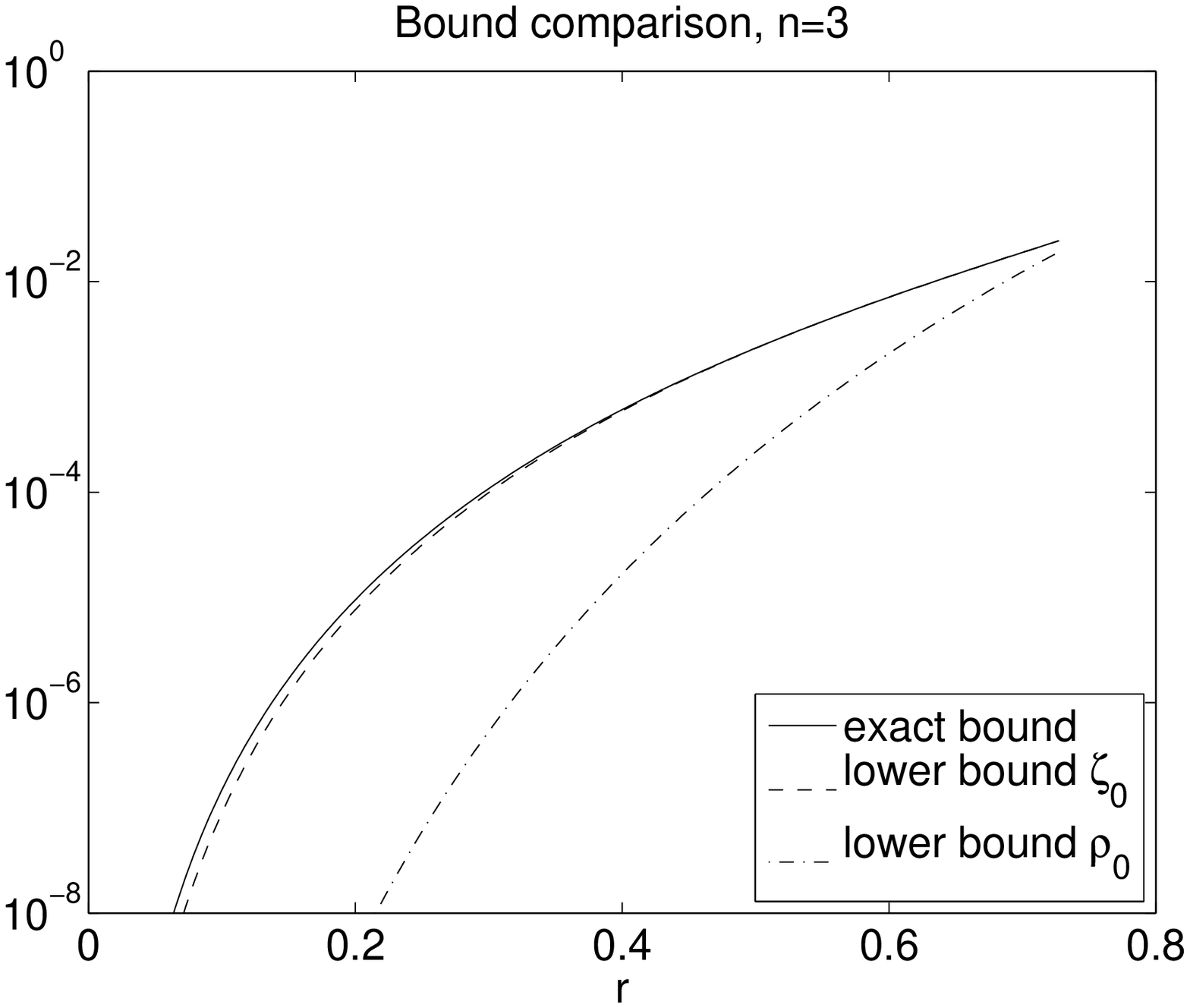}
    \hfill
    \includegraphics[width=.48\textwidth]{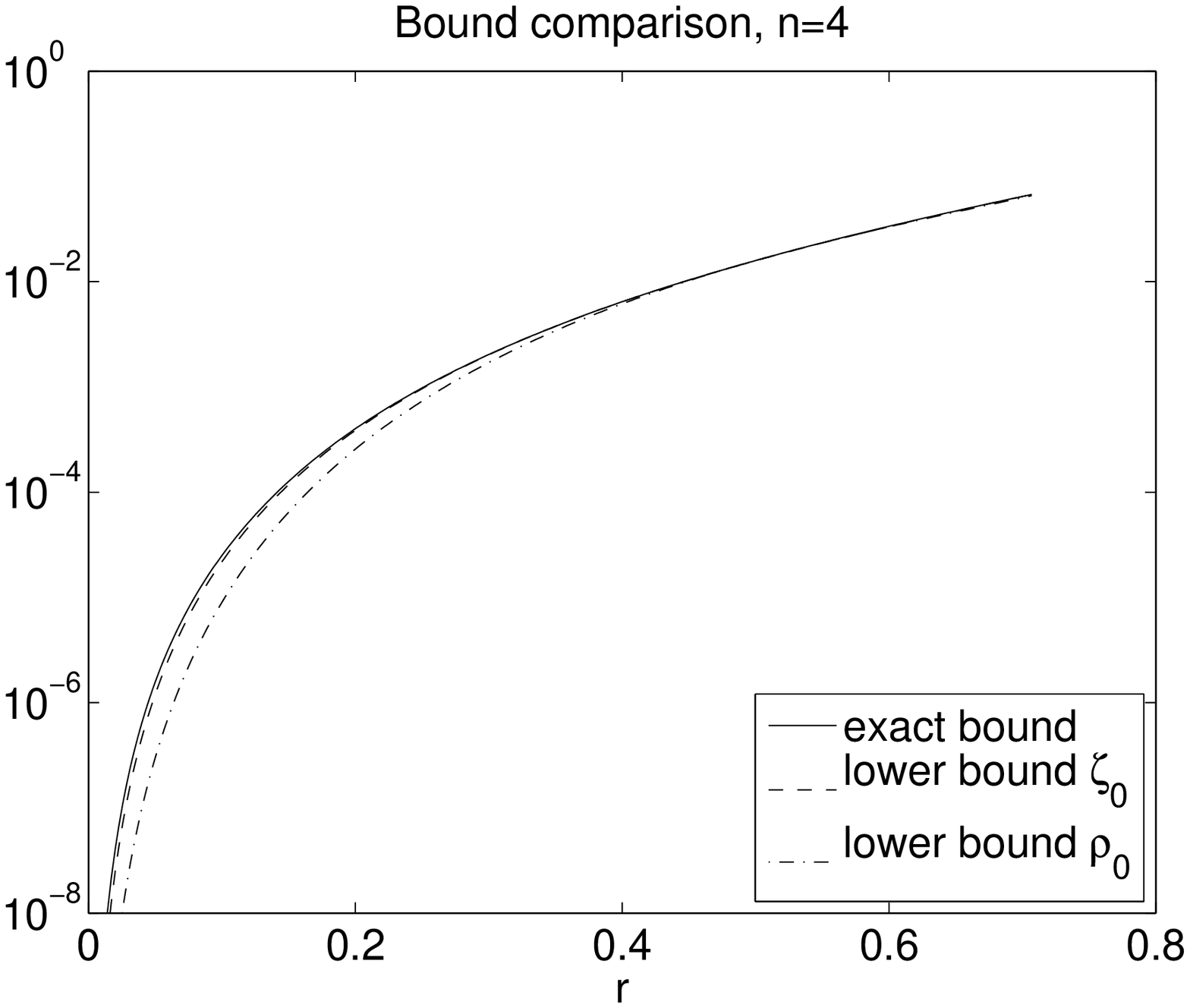}
    \newline
    \includegraphics[width=.48\textwidth]{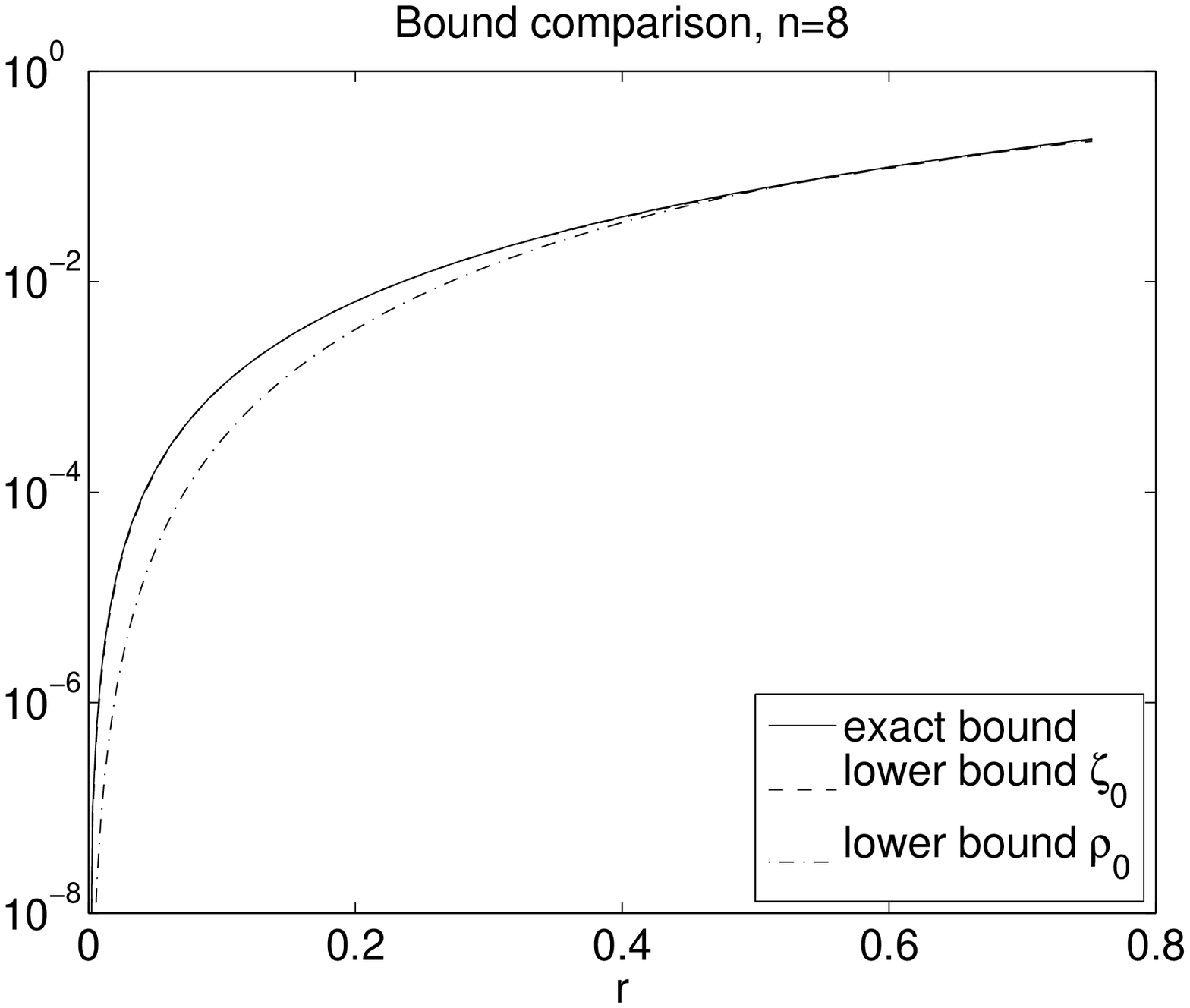}
    \hfill
    \includegraphics[width=.48\textwidth]{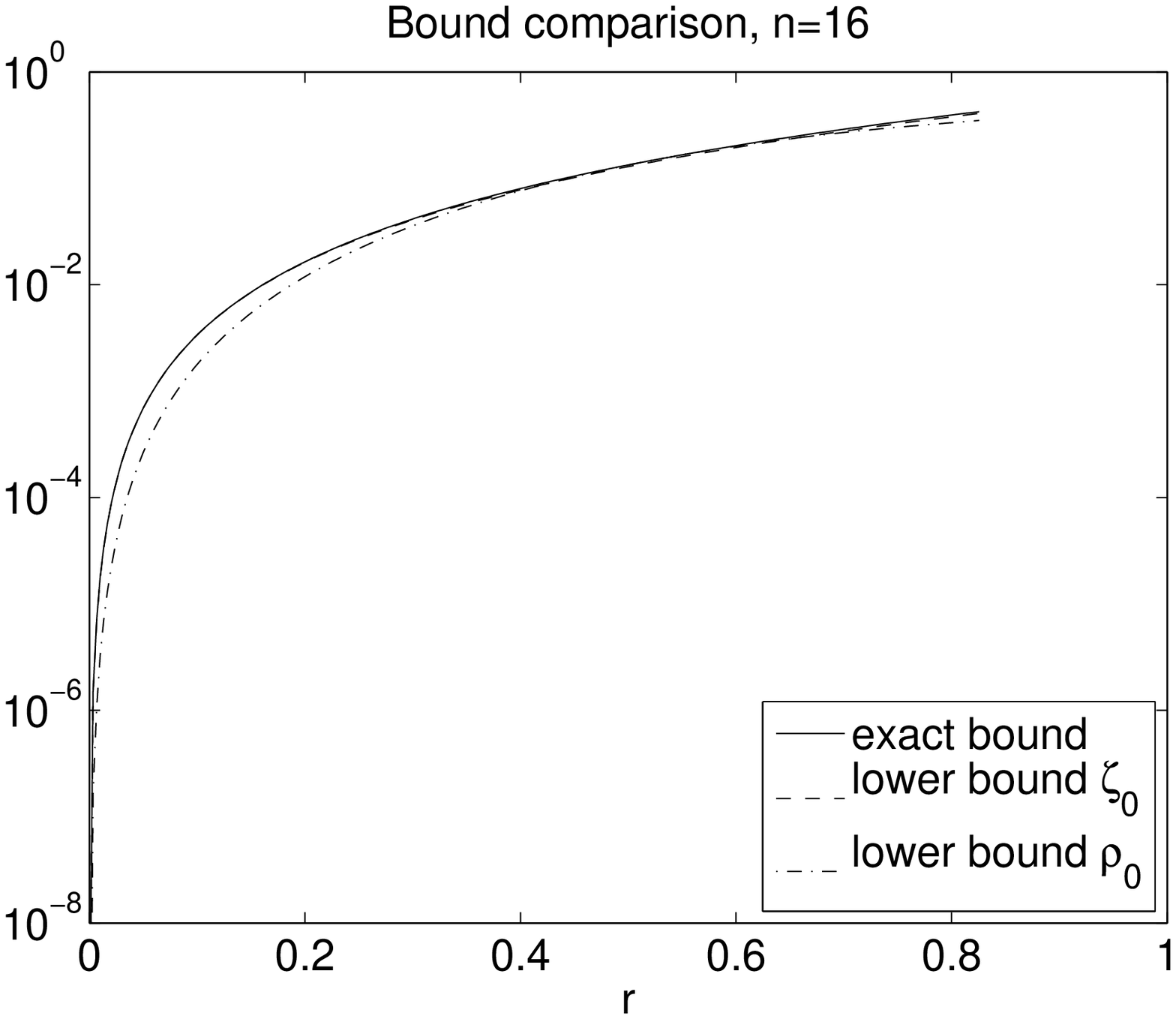}
\end{center}
    \caption{Comparison of the bounds on the eligible masses $\delta$ in
    the equation $S_{s, \delta}(w) = \conj{w}$ such that maximal lensing
    occurs.  For $n=8$ and $n=16$, the exact bound and the lower bound
implied by $\zeta_0$ are visually indistinguishable. Values below
$10^{-8}$ have been clipped.
\label{fig:bd_bound_cmp}}
\end{figure}

In~\cite[equation (7)]{BayerDyer2007}, another bound is derived as
follows.  Set $\rho_0 = s^{1 + 6/n}$, then the authors show that if $s
< 1/\sqrt{2}$ and
\begin{equation}
    \label{eqn:bdmax_rho}
    \delta < \tfrac{\rho_0^{n+2} - \rho_0^n + s^n \rho_0^2}{s^n + \rho_0^n} =
    s^{2 ( \frac{6}{n} + 1 )} - \tfrac{s^6}{1 + s^6},
\end{equation}
the equation $S_{s, \delta}(w) = \conj{w}$ has $5n$ solutions.
Analogous to Remark~\ref{rem:rho0_rhie}, one can show that the bound
is in fact valid for all $s \in (0, r_{\crit})$.

In Fig.~\ref{fig:bd_bound_cmp} we compare for $n=3,4,8$ and $16$ the
exact bound implied by Theorem~\ref{thm:maxlens_bd} with the
bounds~\eqref{eqn:bdmax_zeta} and~\eqref{eqn:bdmax_rho}.  We
see that both bounds are asymptotically close to the exact bound, and
that the bound based on $\zeta_0$ is slightly better than the bound
based on $\rho_0$.

Finally, we note that maximal lensing may occur even for radii $s >
\sqrt{1 + \delta} r_{\crit}$.  We give the following bound beyond which
maximal lensing is impossible.
\begin{cor}
    Let $\delta > 0$ and $s \ge \sqrt{1+\delta}
(\frac{n-2}{n+2})^{\frac{n-2}{2n}}$. Then $S_{s, \delta}(w) = \conj{w}$
has exactly $3n$ solutions.
\end{cor}
\begin{proof}
Recall that $R_{r, \eps}(z) = \conj{z}$ has exactly $3n$ solutions if $r \geq
(\tfrac{n-2}{n+2})^{\frac{n-2}{2n}}$.  Since $s = \tfrac{r}{\sqrt{1-\eps}} =
\sqrt{1+\delta} r$, we see by Proposition~\ref{prop:rhie_bd_trafo} that
$S_{s,\delta}(w) = \conj{w}$ has exactly $3n$ solutions if $s \geq
\sqrt{1+\delta} (\tfrac{n-2}{n+2})^{\frac{n-2}{2n}}$.
\eop
\end{proof}

Bayer and Dyer~\cite{BayerDyer2007,BayerDyer2009} considered the
situation with three equal point masses arranged at the vertices of an
equilateral triangle and a fourth point mass located at its center, i.e.
the case $n=3$ in~\eqref{eqn:bd_perturb}.  As an application of their
bound~\eqref{eqn:bdmax_rho}, they consider a maximal lens with total
deflector mass of $3.5~\Msun$, and derive mass constraints on the
central mass.  Using our \emph{exact} bound on the central mass in
their application (implied by Theorem~\ref{thm:maxlens_bd}) shows
that the central mass is bounded in fact by $0.086~\Msun$ (instead of
$0.07~\Msun$ in~\cite{BayerDyer2007}).

\section{Conclusions and outlook}

We have given a complete characterization of maximal lensing
conditions for Rhie's point lens for radii up to the critical radius
$r_{\crit}$.  We have
presented two bounds on the maximal perturbation mass such that maximal lensing
is guaranteed. The first bound is sharp and computable, but given by the
roots of certain polynomials, while the other bound is explicitly parameterized
and almost as good as the sharp bound.

We have transferred the above mentioned complete characterization to
the lens model of Bayer and Dyer, thus extending the known range of radii for
which maximal lensing occurs.  We obtained an exact bound on the mass for their
model, as well as an approximate bound, which is explicitly given and
which improves the previously known one.  Our analysis also shows that
in the model of Bayer and Dyer
maximal lensing is guaranteed to happen even for radii beyond the critical
radius.

The methods applied in this paper are specific for the lens modeled by
equation~\eqref{eqn:unpert_rat_fun}. Recently we have studied general
perturbations of rational harmonic functions of the form
$f(z)-\conj{z}$, with $f(z)$ rational, by poles~\cite{SeteLuceLiesen2014}.
Our results obtained in~\cite{SeteLuceLiesen2014} may also be of interest in the
study of the gravitational lensing model considered here. Moreover, we
point out that modifications of this model by replacing the point
deflectors by spherically symmetric distributed masses, or allowing
external shear, have been considered, e.g.,
in~\cite{AnEvans2006,BayerDyerGiang2006,KhavinsonLundberg2011,PettersWitt1996}. 
Whether sharp parameter bounds analo\-gous to ours can be obtained in these
models remains a subject of further work.

%% file: proof_rhie_appendix.tex
\section{Auxiliary proofs}

\begin{lem}
\label{lem:zeta_eta}
Let $n \geq 3$ and $r \in (0, r_{\crit})$ be given.  Then
$\zeta(r) \coloneq \tfrac{n+6}{n+8} r^{ \frac{3n+1}{3n-6} }
< (\tfrac{n-2}{n+2})^{\frac{1}{2}}$
holds.
\end{lem}

\begin{proof}
Note that $\zeta(r)$ increases monotonically in $r$, so that
$\zeta(r) < (\tfrac{n-2}{n+2})^{\frac{1}{2}}$ for all $r \in (0, r_{\crit})$ is
equivalent to $\zeta(r_{\crit}) \leq (\tfrac{n-2}{n+2})^{\frac{1}{2}}$.  This
last inequality is equivalent to $(\tfrac{n+6}{n+8})^2 r_{\crit}^{
\frac{6n+2}{3n-6} } \leq \tfrac{n-2}{n+2}$.  Since $\tfrac{n+6}{n+8} < 1$ it
is sufficient to show $r_{\crit}^{6n+2} \leq (\tfrac{n-2}{n+2})^{3n-6}$.
Inserting the value of $r_{\crit}$ yields
$( \tfrac{n-2}{n} )^{3n+1} ( \tfrac{2}{n-2} )^{6+\frac{2}{n}} \leq (
\tfrac{n-2}{n+2} )^{3n-6}$, and after some algebraic manipulation one obtains
\begin{equation}
(1-\tfrac{2}{n}) (n-2)^{-\frac{2}{n}} 2^6 4^{\frac{1}{n}}
(1+\tfrac{2}{n})^{3n-6} \leq n^{6}. \label{eqn:n_bound_for_zeta}
\end{equation}
The left hand side can roughly be bounded as follows:
\begin{equation*}
(1-\tfrac{2}{n}) (n-2)^{-\frac{2}{n}} \cdot 2^6 \cdot 4^{\frac{1}{n}}
(1+\tfrac{2}{n})^{3n-6}
\leq 1 \cdot 1 \cdot 2^6 \cdot 4^{\frac{1}{3}} \big( (1+\tfrac{2}{n})^n \big)^3
\leq 4^{\frac{1}{3}} (2e)^6.
\end{equation*}
Now, $4^{\frac{1}{3}} (2e)^6 \leq n^6$ is equivalent to $n \geq 4^{\frac{1}{18}}
2 e \approx 5.8$, so \eqref{eqn:n_bound_for_zeta} holds for $n \geq 6$.
It is not difficult but tiresome to see by direct calculation that
\eqref{eqn:n_bound_for_zeta} is also valid for $n = 3, 4, 5$.
\eop
\end{proof}

\begin{thm}
    \label{thm:2}
Let $n \geq 3$, $\delta > 0$ and $s \in (0, \sqrt{1+\delta} r_{\crit})$ be
given and denote by $\xi_1 \defby \xi_1(\delta)$ be the smallest
positive root of the polynomial $(n+2) \xi^n - (1+\delta) n \xi^{n-2} +
2 s^n$.  Then the equation $S_{s,\delta}(w) = \conj{w}$ has $5n$ solution
if and only if
\begin{equation*}
    \delta < \tfrac{\xi_1^{n+2} - \xi_1^n + s^n \xi_1^2}{s^n + \xi_1^n}
    \bydef u(\delta).
\end{equation*}
If $\delta < u(\delta)$, the $5n$ solutions have the form $s_2
e^{i \frac{(2k+1)\pi}{n}}$, $s_3 e^{i \frac{(2k+1)\pi}{n}}$, $s_4 e^{i
\frac{(2k+1)\pi}{n}}$, and $s_1 e^{i \frac{2k\pi}{n}}$, $s_5 e^{i
\frac{2k\pi}{n}}$, $k=0,1,\dots,n-1$, where $0 < s_1 < \sqrt{\delta} <
s_2 < s_3 < s_4 < \sqrt{1+\delta} < s_5$.  If $\delta = u(\delta)$,
there exist $4n$ solutions, and only $3n$
solutions exist for $\delta > u(\delta)$.
\end{thm}

\begin{proof}
Obviously, $w=0$ does not solve $S_{s,\delta}(w) = \conj{w}$. As in the
proof of Proposition~\ref{prop:num_sol}, we now write $w = \rho e^{i
\varphi}$, where $\rho > 0$ and $\varphi \in \R$. Then $S_{s,\delta}(w) =
\conj{w}$ can be written as
\begin{equation}
\label{eqn:comp_rho_phi_pert_bd}
    (1 + \delta - \rho^2) \rho^n e^{i n \varphi} = (\delta - \rho^2) s^n.
\end{equation}
Since the right hand side of \eqref{eqn:comp_rho_phi_pert_bd} is real
and $s>0$, it is necessary that $\rho \neq \sqrt{1+\delta}$, $\rho\neq
\sqrt{\delta}$, and that either $e^{i n \varphi} = 1$ or $e^{i n
\varphi} = -1$.

If $e^{i n \varphi} = 1$, then $\varphi = \tfrac{2k\pi}{n}$ for some
$k\in\{0,1,\dots,n-1\}$. Furthermore, $1 + \delta - \rho^2$ and $\delta -
\rho^2$ must have same sign, hence either $0 < \rho < \sqrt{\delta}$ or
$\sqrt{1 + \delta} < \rho$, and \eqref{eqn:comp_rho_phi_pert_bd} can be
written as $f_+(\rho) \coloneq \rho^{n+2} - (1 + \delta) \rho^n - s^n
\rho^2 + \delta s^n = 0$.

By Descartes' rule of signs $f_+(\rho)$ has either zero or two
positive roots.  From
\begin{equation*}
    f_+(0) = \delta s^n > 0, \quad
    f_+(\sqrt{\delta}) = - \sqrt{\delta}^n < 0, \quad
    f_+(\sqrt{1+\delta}) = - s^n < 0,
\end{equation*}
and $f_+(\rho)\rightarrow\infty$ for $\rho\rightarrow\infty$, we see
that $f_+(\rho)$ indeed has one root $s_1\in (0, \sqrt{\delta})$ and one
root $s_5\in (\sqrt{1+\delta}, \infty)$.  Consequently, for all $s>0$
and all $\delta > 0$ there exist $2n$ solutions $s_1 e^{i \frac{2 k
\pi}{n}}$ and $s_5 e^{i \frac{2 k \pi}{n}}$, $k=0,1,\dots,n-1$, of the
equation $S_{s,\delta}(w)=\conj{w}$.

Next consider the case $e^{i n \varphi} = -1$, then $\varphi =
\tfrac{(2k+1)\pi}{n}$ for some $k\in\{0,1,\dots,n-1\}$. Here $1 + \delta
- \rho^2$ and $\delta - \rho^2$ must have opposite signs. Hence
$\sqrt{\delta} < \rho < \sqrt{1 + \delta}$ is necessary, and
\eqref{eqn:comp_rho_phi_pert_bd} can be written as $f_-(\rho) \coloneq
\rho^{n+2} - (1 + \delta) \rho^n + s^n \rho^2 - \delta s^n = 0.$ The polynomial
$f_-(\rho)$ has either one or three positive roots.
We will derive necessary and sufficient conditions so that $f_-(\rho)$
has three distinct positive roots in the interval
$(\sqrt{\delta},\sqrt{1+\delta})$.

The positive roots of the derivative $f_{-}'(\rho) = \rho \big( (n+2)
\rho^n - n (1+\delta) \rho^{n-2} + 2 s^n \big)$ are equal to the
positive roots of $g(\rho) \coloneq (n+2) \rho^n - n (1+\delta)
\rho^{n-2} + 2 s^n.$ From
\begin{equation*}
    g'(\rho) = n \rho^{n-3} \big( (n+2) \rho^2 - (n-2) (1+\delta) \big),
\end{equation*}
we see that the unique positive root of $g'(\rho)$ is $\eta \coloneq
\sqrt{1+\delta} (\frac{n-2}{n+2})^{\frac{1}{2}} \in (0,\sqrt{1+\delta})$.
One can calculate that for all $s\in (0,\sqrt{1+\delta}r_{\crit})$ we
have $g(\eta) < 0$.  Together with $g(0) = 2 s^n > 0$ and
$g(\sqrt{1+\delta}) = 2 ( \sqrt{1+\delta}^n + s^n) > 0$ this shows that
$g(\rho)$ and thus $f_-'(\rho)$ have exactly two positive roots, say
$\xi_1$ and $\xi_2$, with $0 < \xi_1 < \eta < \xi_2 < \sqrt{1 +
\delta}$.  Note that $\xi_1$ and $\xi_2$ depend on $\delta$.

Let us write $f_{-}(\rho) = \rho^2 p(\rho) - \delta s^n$, where $p(\rho)
\coloneq \rho^n - (1 + \delta) \rho^{n-2} + s^n$.  It is easy to see
that $p(\rho)$ has two positive roots if and only if $s \in (0,
\sqrt{1+\delta} r_{\crit})$.  These zeros $z_1, z_2$ satisfy $0 < z_1 <
\sqrt{1+\delta} (\frac{n-2}{n})^{\frac{1}{2}} < z_2 < \sqrt{1+\delta}$.

Since $f_{-}(0) = f_{-}(z_1) = f_{-}(z_2) = - \delta s^n$, the mean
value theorem implies that the two roots of $f_{-}'(\rho)$ satisfy $0
< \xi_1 < z_1 < \xi_2 < z_2$.  From $f_{-}(z_2) < 0 <
f_{-}(\sqrt{1+\delta})$ we then see that $f_-(\rho)$ has exactly one
root $s_4\in (z_2, \sqrt{1+\delta})$.

Further, $f_{-}(0) = f_{-}(z_1) < 0$ implies that $f_-(\rho)$ has two
more (distinct) roots if and only if $f_{-}(\xi_1) > 0$, or,
equivalently,
\begin{equation*}
    \delta <
    \tfrac{ \xi_1^{n+2} - \xi_1^n + \xi_1^2 s^n}{s^n + \xi_1^n} = u(\delta).
\end{equation*}

If the above inequality is satisfied, the function $f_-(\rho)$ has
two more distinct roots $s_2, s_3$ with $\sqrt{\delta} < s_2 < \xi_1 <
s_3 < z_1$.  Hence, there exist $3n$ solutions $s_2 e^{ i \frac{(2k+1)
\pi}{n}}$, $s_3 e^{ i \frac{(2k+1) \pi}{n}}$ and $s_4 e^{ i
\frac{(2k+1) \pi}{n}}$, $k=0,1,\dots,n-1$, of the equation
$S_{s,\delta}(w) = \conj{w}$.

On the other hand, if $\delta = u(\delta)$, $\xi_1$ is a (double) zero of
$f_-(\rho)$.  Then $S_{s,\delta}(w) = \conj{w}$ has the $2n$ solutions
$\xi_1 e^{i \frac{(2k+1) \pi}{n}}$ and $s_4 e^{i \frac{(2k+1)
\pi}{n}}$, $k = 0, 1, \ldots, n-1$ in addition to the $2n$ solutions
corresponding to $s_1$ and $s_5$.  Finally, if $\delta > u(\delta)$,
then only the $n$ solutions $s_4 e^{ i \frac{(2k+1)
\pi}{n}}$, $k=0,1,\dots,n-1$, of $S_{s,\delta}(w) = \conj{w}$ occur in
addition to the $2n$ roots corresponding to $s_1$ and $s_5$.  \eop
\end{proof}
